\newif\ifarxiv
\newif\ifappendixmode
\newif\iflogicprograms

\arxivtrue 

\logicprogramsfalse 

\documentclass{article}
\usepackage{ijcai24}

\usepackage{times}
\usepackage{soul}
\usepackage{url}
\usepackage[hidelinks]{hyperref}
\usepackage[utf8]{inputenc}
\usepackage[small]{caption}
\usepackage{graphicx}
\usepackage{amsmath}
\usepackage{amsthm}
\usepackage{booktabs}
\usepackage{algorithm}
\usepackage{algorithmic}
\usepackage[switch]{lineno}

\ifarxiv
\hypersetup{
  breaklinks=true   }
\fi

\usepackage{amssymb}

\graphicspath{{images/}}

\usepackage{ltl}
\usepackage{ifdraft}
\usepackage{thm-restate}
\usepackage{mymacros}
\usepackage{mathtools}
\usepackage{enumitem}

\newtheorem{example}{Example}
\newtheorem{theorem}{Theorem}

\newtheorem{definition}{Definition}

\newtheorem{proposition}{Proposition}

\newtheorem{corollary}{Corollary}

\newcommand{\sinceop}{\LTLs}

\newcommand{\beforeop}{\LTLcircleminus}
\newcommand{\prevop}{\LTLcircleminus}
\newcommand{\palwaysop}{\LTLsquareminus}
\newcommand{\peventuallyop}{\LTLdiamondminus}

\iflogicprograms
  \newcommand{\ruleifinline}{\operatorname{\coloneqq}}
\else
  \newcommand{\ruleifinline}{\operatorname{\coloneqq}}
\fi
\newcommand{\ruleif}{\,\ruleifinline\,}

\newcommand{\charsemigroup}{\operatorname{\mathbf{S}}}

\iflogicprograms
  \newcommand{\wrule}{definition}
  \newcommand{\wbody}{body}
  \newcommand{\wprogram}{program}

\else
  \newcommand{\wrule}{definition}
  \newcommand{\wbody}{body}
  \newcommand{\wprogram}{program}

\fi

\iflogicprograms

\else

\fi

\title{The Transformation Logics}

\author{
    Alessandro Ronca
    \affiliations
    University of Oxford
    \emails
    alessandro.ronca@cs.ox.ac.uk
}

\begin{document}

\maketitle

\begin{abstract}
  We introduce a new family of temporal logics designed to finely balance the
  trade-off between expressivity and complexity.
  Their key feature is
  the possibility of defining operators of a new kind that we call
  \emph{transformation operators}.
  Some of them subsume existing temporal operators, while others are entirely
  novel. 
  Of particular interest are transformation operators
  \emph{based on semigroups}. They enable logics to harness the richness of semigroup
  theory, and we show them to yield logics capable of
creating \emph{hierarchies of increasing expressivity and complexity}
which are non-trivial to characterise in existing logics.
  The result is a genuinely novel and yet unexplored 
  \emph{landscape of temporal logics},
  each of them with the potential of matching the 
  trade-off between expressivity and complexity required by specific
  applications.
\end{abstract}

\section{Introduction}

We introduce the \emph{Transformation Logics}, a new family of temporal logics
designed to finely balance the trade-off between expressivity and complexity.
Their key feature is the possibility of defining operators of a new kind that we
call \emph{transformation operators}. They capture patterns over sequences, and
they can be thought of as a generalisation of temporal operators.
The subclass of
transformation operators based on \emph{finite semigroups} is of particular
interest.
Such \emph{semigroup-like} operators suffice to capture all regular languages,
and remarkably
they allow for creating 
\emph{hierarchies of increasing expressivity and complexity} which
are non-trivial to define in existing logics. 
The base level of such hierarchies is obtained using the operator defined by
the \emph{flip-flop monoid}. The other levels are obtained introducing operators
based on \emph{simple groups}---the building blocks of all groups.
Simple groups have been systematically classified into a finite number of
families, cf.\ \cite{gorenstein2018classification}.
The classification provides a compass in the landscape of groups, and a roadmap
in the exploration of temporal logics, as it is made clear by the results in
this paper.

Our motivation arises from the usage of temporal logics in AI. 
They are used in \emph{reinforcement learning} to specify reward and
dynamics functions
\cite{bacchus1996rewarding,brafman2018ltlf,toroicarte2018teaching,camacho2019ltl,degiacomo2020bolts,degiacomo2020temporal};
in \emph{planning} for describing temporally-extended goals
\cite{torres2015polynomial,camacho2017nondeterministic,degiacomo2018automata,brafman2019planning,icaps2023bdffgs};
in \emph{stream reasoning} to express programs with the ability of referring
to different points of a stream of data
\cite{beck2018lars,ronca2022delay,walega2023stream}.

In the above applications, the required trade-off between expressivity and
complexity will depend on the case at hand. When the basic
expressivity of the star-free regular languages suffices, one can
employ logics such as Past LTL \cite{manna1991completing} and LTLf
\cite{degiacomo2013ltlf}.
In all the other cases, one needs to resort to more expressive logics.
The existing extensions of the above logics have the expressivity of all regular
languages, cf.\ ETL \cite{wolper1983temporal} and LDLf \cite{degiacomo2013ltlf}.
This is a big leap in expressivity,
which may incur an unnecessarily high computational complexity. 
We show next two examples where the required expressivity lies in fragments
between the star-free regular languages and all regular languages. 
These 
\ifarxiv
intermediate 
\fi
fragments can be precisely characterised in the Transformation Logics.

\begin{example}
  An agent is assigned a task that can be completed multiple times.
  We receive an update every minute telling us whether the agent has completed
  the task in the minute that has just elapsed.
  We need to detect whether the agent has completed the task at least
  once on every past day.
\end{example}

The example describes a periodic pattern, which is beyond the star-free
regular languages. It requires to count minutes modulo $24*60 = 1440$ in order
to establish the end of every day. This can be expressed in the Transformation
Logics using a transformation operator defined by the cyclic group
$C_{1440}$.
\ifarxiv
  Or alternatively using three transformation operators defined the
  cyclic groups $C_2$, $C_3$, and $C_5$, respectively. 
\fi
Cyclic group operators yield an ability to capture many useful periodic
patterns. At the same time, they belong to the special class of 
\emph{solvable group operators}, which enjoys good properties such as a
more favourable computational complexity compared to larger classes of
operators. 
\ifarxiv
\par
\fi
The next one is an example where solvable group operators do not
suffice, and we need to resort to \emph{symmetric group operators}.

\begin{example}
  A cycling race with $n$ participants takes place, and we need to keep track of
  the live ranking.
  At each step an overtake can happen, in which case it is communicated to us in
  the form $(i,i+1)$ meaning that the cyclist in position $i$ has overtaken the
  one in position $i+1$.
  We know the initial ranking, and we need to keep track of the live ranking.
\end{example}

The ranking in the example corresponds to the symmetric group $S_n$,
which is not solvable if $n \geq 5$.
The example can be specified in a Transformation Logic featuring a
transformation operator defined by the group $S_n$.

\paragraph{Summary of the contribution.}
We introduce the Transformation Logics, providing a formal syntax and
semantics.
Their main characteristic is the transformation operators. The operators are
very general, as we demonstrate through a series of concrete examples. 
We develop a systematic approach in defining operators, based on
semigroup theory and algebraic automata theory, cf.\
\cite{ginzburg,arbib1969theories,domosi2005algebraic}. This way we are able to
identify \emph{prime operators} that can capture all finite operators.
For them, we prove a series of expressivity and complexity results.
Regarding the \emph{expressivity}, we show there exists one operator, defined by
the \emph{flip-flop monoid}, which yields the expressivity of the
\emph{star-free regular languages}; as one keeps adding operators based on
\emph{cyclic groups} of prime order, the expressivity increases, up to capturing
all languages that can be captured using \emph{solvable group operators}; the
expressivity of all regular languages is
reached by adding the other prime operators, that can be defined by choosing
groups from \emph{the classification of finite simple groups}, cf.\
\cite{gorenstein2018classification}.
Regarding the \emph{complexity}, we focus on the \emph{evaluation problem}, and
we show three sets of results.
First, we show any Transformation Logic can be evaluated in \emph{polynomial
time}, whenever its operators can be evaluated in polynomial time. 
Second, for two notable families of operators, we show that
polynomial-time evaluation is possible even when they are represented compactly.
Third, we focus on the \emph{data complexity} of evaluation showing that it
corresponds
to the three circuit complexity classes $\textsc{AC}^0 \subsetneq \textsc{ACC}^0
\subseteq \textsc{NC}^1$ when we include (i)~only the flip-flop operator,
(ii)~also cyclic operators, and (iii)~all operators. 
Finally, we show how Past LTL formulas can be easily translated into the core
Transformation Logic featuring the flip-flop operator. 
\ifarxiv
  \paragraph{Appendix.}
  This paper is accompanied by an appendix which contains: 
  proofs of all our results; details of the examples; a discussion on how Past
  LTL can be extended with transformation operators.
\fi

\section{Preliminaries}

For $X$ a set, a \emph{transformation} is a function $f: X \to X$.
We write the identity function over any domain as $\mathit{id}$.
We denote the Boolean domain $\{ 0,1 \}$ by $\mathbb{B}$, the natural numbers
by $\mathbb{N}$, and the integer numbers by $\mathbb{Z}$.

\subsection{Formal Languages}
An \emph{alphabet} $\Sigma$ is a non-empty finite set of elements called
\emph{letters}. A \emph{string} over $\Sigma$ is a finite concatenation
$\sigma_1 \cdots \sigma_n$ of letters from $\Sigma$. A \emph{language} over
$\Sigma$ is a set of strings over $\Sigma$.
The \emph{regular languages} are the ones languages that are defined by regular
expressions, or equivalently by finite automata
\cite{kleene1956representation}.
The \emph{star-free regular languages} are the languages that are defined
by regular expressions without the Kleene star but with complementation,
or equivalently by a group-free finite automaton, cf.\ \cite{ginzburg}.

\subsection{Propositional Logic}

\par
\noindent
\textbf{Syntax.}
A \emph{propositional variable} is an element from a set $\mathcal{V}$ that we
consider as given. 
Typically we denote propositional variables by lowercase
Latin letters.
A \emph{propositional formula} is built out of propositional variables and the
\emph{Boolean operators} $\{ \neg, \land, \lor \}$. It is defined inductively as
a propositional variable or one of the following expressions: 
$\neg \alpha$, $\alpha \land \beta$, $\alpha \lor \beta$ where $\alpha$ and
$\beta$ are propositional formulas.
Additional Boolean operators may be defined, but it is not necessary as the
former operators are \emph{universal}, they suffice to express all Boolean
functions. 

\medskip
\par
\noindent
\textbf{Semantics.}
An \emph{interpretation $I$ for a propositional formula} is a subset of
the propositional variables occurring in the formula.
Intuitively, 
the fact that a variable appears in the intepretation means that
the variable stands for a proposition that is true.
\ifarxiv
An interpretation can also be seen as an assignment. 
\fi
An \emph{assignment} 
is a function $\nu : V \to \mathbb{B}$ from a set of propositional
variables $V$ to the Boolean domain 
$\mathbb{B} = \{ 0,1 \}$.
When $V = \{ v_1, \dots, v_n \}$,
we can also write an assingment $\nu$ as the map
$\langle v_1, \dots, v_n \rangle \mapsto \langle b_1, \dots, b_n \rangle$, with
the meaning that $\nu(v_i) = b_i$.
An interpretation $I$ corresponds to the assignment $\nu$ such that 
$\nu(a) = 1$ iff $a \in I$.
Then, the semantics of formulas is defined in terms of the following
satisfiability relation.
Given a formula $\alpha$ and an interpretation $I$ for $\alpha$,
the \emph{satisfiability relation} $I \models \alpha$ is defined following the
structural definition of formulas, for variables as
\begin{itemize}
  \item
    $I \models a$ iff $a \in I$,
\end{itemize}
and inductively for the other formulas as
\begin{itemize}
  \item
    $I \models \neg \alpha$ iff $I \not\models \alpha$,
  \item
    $I \models \alpha \lor \beta$ iff $I \models \alpha$ or $I \models \beta$,
  \item
    $I \models \alpha \land \beta$ iff $I \models \alpha$ and $I \models \beta$.
\end{itemize}
An assignment $\nu$ to variables $\{ a_1, \dots, a_m \}$ can be
seen as the conjunction $l_1 \land \dots \land l_m$ where $l_i = a_i$ if 
$\nu(a_i) = 1$ and $l_i = \neg a_i$ otherwise. 
This allows us to write $I \models \nu$.

\subsection{Semigroups and Groups}

A \emph{semigroup} is a non-empty set together with an
\emph{associative} binary
operation that combines any two elements $a$ and $b$ of the set to form a third
element $c$ of the set, written $c = (a \cdot b)$.
A \emph{monoid} is a semigroup that has an \emph{identity element} $e$, i.e., 
$(a \cdot e) = (e \cdot a) = a$ for every element $a$. 
The identity element is unique when it exists. 
A \emph{group} is a monoid where every element $a$ has an \emph{inverse}
$b$, i.e., $(a \cdot b) = (b \cdot a) = e$ where $e$ is the identity element.
For every element $a$ of a group, its inverse is unique and it is denoted by
$a^{-1}$.
A \emph{subsemigroup} (\emph{subgroup}) of a semigroup $S$ is a subset of $S$
that is a semigroup (group).
The \emph{order} of a semigroup is the number of elements.
For $S$ and $T$ semigroups, we write 
$ST = \{ s \cdot t \mid s \in S\text{, } t \in T\}$; we also write 
$S^1 = S$ and $S^n = SS^{n-1}$.
A semigroup $S$ is \emph{generated} by a semigroup $T$ if 
$S = \bigcup_n T^n$.
A \emph{homomorphism} from a semigroup $S$ to a semigroup $T$ is a mapping 
$\psi:S \to T$ such that $\psi(s_1 \cdot s_2) = \psi(s_1) \cdot \psi(s_2)$ for
every $s_1,s_2 \in S$.
If $\psi$ is bijective, we say that $S$ and $T$ are \emph{isomorphic}. 
Isomorphic semigroups are considered identical.
Let $G$ be a group and let $H$ be a subgroup of $G$. 
A \emph{right coset} of $G$ is
$g \cdot H = \{  g \cdot h \mid h \in H\}$ for $g \in G$, and a \emph{left
coset} of $G$ is $H \cdot g = \{  h \cdot g \mid h \in H\}$ for $g \in G$. 
\ifarxiv
Subgroup $H$ is \emph{normal} if its left and
right cosets coincide, i.e., $g \cdot H = H \cdot g$ for every $g \in G$.
\fi
A group is \emph{trivial} if it is the singleton $\{ e \}$.
A \emph{simple group} is a group $G$ such that every normal subgroup of $G$ is
either trivial or $G$ itself. 
For $g,h \in G$, the \emph{commutator} of $g$ and $h$ is $g^{-1} \cdot h^{-1}
\cdot g \cdot h$.
The \emph{derived subgroup} of $G$ is the subgroup generated by its
commutators.
Setting $G^{(0)} = G$, 
the $n$-th derived subgroup
$G^{(n)}$ is the derived subgroup of $G^{(n-1)}$.
A group is \emph{solvable} if there exists $n$ such that $G^{(n)}$ is trivial.

A \emph{flip-flop monoid} is a three-element monoid $\{s,r,e\}$
where $(r \cdot s) = s$, $(s \cdot s) = s$,
$(r \cdot r) = r$, and $(s \cdot r) = r$.
All flip-flop monoids are isomorphic, and hence one refers to \emph{the}
flip-flop monoid.
A \emph{cyclic group} is a group that is isomorphic to the group $C_n$ of
integers 
$\{ 0, \dots, n-1 \}$ 
with modular addition 
$i \cdot j = i + j \operatorname{mod} n$.
Again, one refers to any cyclic group of order $n$ as \emph{the} cyclic group
$C_n$. Two relevant properties of cyclic groups are:
(i)~a cyclic group $C_n$ is simple iff $n$ is a prime number;
(ii)~every cyclic group is solvable.

\section{The Transformation Logics}

We introduce the Transformation Logics.
They are a propositional formalism, atoms are variables standing for
propositions that can be true or false.
The truth of some variables is given as input, whereas the meaning of other
variables is given by a \emph{definition}.
Definitions allow us to avoid nested expressions, and hence they aid
intelligibility in this context.
A definition features either a Boolean expression, the delay operator, or a
transformation operator. The delay operator is akin to the before operator from
Past LTL. A transformation operator has a domain of elements and a set of
transformations over the domain---a \emph{transformation} is a map from elements
of a domain to elements of the same domain.
At every step a transformation operator has an associated domain element to
which it applies a transformation based on the truth value of the operands. 
Then the evaluation value is a function of the current domain element. 
Each transformation corresponds to a specific functionality. For example,
setting a bit to one, or increasing a count. 
More intuition is given below in the Paragraph `Explanation' and later in
Section~\ref{sec:example-operators}.

\medskip
\par
\noindent
\textbf{Syntax.}
A \emph{static definition} is 
\begin{align*}
p \ruleif \alpha
\end{align*}
where $p$ is a propositional variable, and $\alpha$ is a propositional formula. 
A \emph{delay definition} is 
\begin{align*}
p \ruleif \operatorname{\mathcal{D}} q
\end{align*}
where $p$ and $q$ are propositional variables, and $\mathcal{D}$ is called the
\emph{delay operator}.
A \emph{transformation operator} $\mathcal{T}$ is a tuple  
$\langle X, T, \phi, \psi \rangle$ consisting of
a non-empty set $X$,
a non-empty set $T$ of transformations $\tau: X \to X$, 
a function $\phi: \mathbb{B}^m \to T$, and 
a function $\psi: X \to \mathbb{B}^n$.
We call $X$ the \emph{transformation domain};
we call $m$ and $n$ the \emph{input} and \emph{output arity}, respectively.
An operator is \emph{finite} if its transformation domain is finite---in which
case all its other components are necessarily finite.
A \emph{transformation definition} is 
\begin{align*}
p_1, \dots, p_n \ruleif \operatorname{\mathcal{T}}(q_1, \dots, q_m \mid x_0)
\end{align*}
where 
$\mathcal{T}$ is a transformation operator with input arity $m$ and output arity
$n$; $p_1,\dots,p_n$ and $q_1, \dots, q_m$ are variables; and $x_0 \in X$ is the
initial domain element.
A \emph{definition} is either a static definition, a delay definition, or a
transformation definition.
In a definition, the expression on the left of `$\ruleif$' is called
\emph{head}, and the expression on the right is called \emph{body}.
A \emph{program} is a finite set of definitions.
A \emph{query} is a pair $(P,q)$ consisting of a program $P$ and a variable $q$
occurring in $P$.
Programs are required to be \emph{nonrecursive}, i.e., to have an acyclic
dependency graph.
The \emph{dependency graph} of a program has one node for each variable
in the program, and it has a directed edge from $a$ to $b$ if there is a
definition where $a$ is in the body and $b$ is the head.
Programs are also required to define variables at most once, i.e., every
variable $p$ occurs at most once in the head of a definition; if $p$ occurs in
the head of a definition $d$, we say that $d$ \emph{defines} $p$.
In a given program $P$, a variable is a \emph{defined variable} if there is a
definition in $P$ that defines it, and it is an \emph{input variable}
otherwise.

For $\mathbf{T}$ a set of transformation operators,
the \emph{Transformation Logic $\mathcal{L}(\mathbf{T})$} is the set of
programs consisting of all static definitions, all delay definitions, and all
transformation definitions with operators from $\mathbf{T}$.

\medskip
\par
\noindent
\textbf{Semantics.}
An \emph{input} to a program $P$ is a finite non-empty sequence of subsets of
the input variables of $P$.
An \emph{assignment to a transformation definition} $d$ is an expression 
$d \mapsto x$ with $x$ a domain element of the operator of $d$.
The semantics of programs is defined in terms of the following
\emph{satisfiability relation}.
Given a program $P$, an input $I = I_1, \dots, I_\ell$ to $P$, and an index 
$t \in [1,\ell]$, we define the satisfiability relation $(P, I, t) \models E$
where $E$ is a propositional variable, a propositional formula, or an
assignment to a transformation definition.
We assume definitions are in the form given above; it allows us to refer
to the symbols mentioned there, e.g., symbol $p_i$ for the $i$-th head variable
of a transformation definition.
\begin{enumerate}[leftmargin=*]
  \item
    For $a$ an input variable,
    \begin{itemize}[leftmargin=*]
      \item
        $(P, I, t) \models a$ iff $a \in I_t$;
    \end{itemize}
  \item
    For $\alpha$ and $\beta$ formulas,
    \begin{itemize}[leftmargin=*]
      \item
        $(P, I, t) \models \neg \alpha$ iff $(P, I, t) \not\models \alpha$;
      \item
        $(P, I, t) \models \alpha \lor \beta$ iff $(P, I, t) \models \alpha$ or
        $(P, I, t) \models \beta$;
      \item
        $(P, I, t) \models \alpha \land \beta$ iff $(P, I, t) \models \alpha$
        and $(P, I, t) \models \beta$;
    \end{itemize}
  \item
    For $p$ a variable defined by a static definition,
    \begin{itemize}[leftmargin=*]
      \item
      $(P, I, t) \models p$ iff $(P, I, t) \models \alpha$;
  \end{itemize}
\item
  For $p$ a variable defined by a delay definition,
  \begin{itemize}[leftmargin=*]
    \item
      $(P, I, t) \models p$ iff $(P, I, t-1) \models q$;
  \end{itemize}
  \item
    For $d \mapsto x$ an assignment to a transformation definition $d$,
    \begin{itemize}[leftmargin=*]
      \item
        $(P,I,0) \models (d \mapsto x)$ iff $x = x_0$;
      \item
        $(P,I,t) \models (d \mapsto x)$ iff 
        \begin{itemize}[leftmargin=*]
          \item
            $(P,I,t-1) \models (d \mapsto x')$,
          \item
            $(P,I,t) \models (\langle q_1, \dots, q_m \rangle \mapsto \mu)$, and
          \item
            $\tau(x') = x$ with $\tau = \phi(\mu)$;
        \end{itemize}
    \end{itemize}
\item
  For $p_i$ a variable defined by a transformation definition $d$,
  \begin{itemize}[leftmargin=*]
    \item
  $(P,I,t) \models p_i$ iff 
    \begin{itemize}[leftmargin=*]
      \item
        $(P,I,t) \models (d \mapsto x)$, and
      \item
        $\psi(x) = \langle b_1, \dots, b_i, \dots, b_n \rangle$
        with $b_i = 1$.
    \end{itemize}
  \end{itemize}
\end{enumerate}

\medskip
\par
\noindent
\textbf{Explanation.}
The index $t$ can be thought of as a time point, ranging over the positions of
the given input $I$.
Points~1 and~2 follow the standard semantics of propositional formulas,
evaluated with respect to the assignment $I_t$. 
Points~3--5 define the semantics of definitions, relying on the auxiliary Point~6.
Point~3 is the semantics of static definitions. It says that the truth value of
a variable $p$ defined by a static definitions is the truth value of the
propositional formula $\alpha$ corresponding to the body of the definition,
evaluated at the same time point.
Point~4 is the semantics of delay definitions. It says that the truth value of a
variable $p$ defined by a delay definition is the truth value of the
propositional variable $q$ occurring in the body the definition, evaluated
at the \emph{previous} time point.
Point~5 describes the element $x$ that is currently associated with a
transformation definition.
For $t=0$, the element is $x_0$, the initial element specified in the
definition. 
For $t > 0$, the element $x$ is determined as follows.
We pick an assignment $\langle q_1, \dots, q_m \rangle \mapsto \mu$ for the
body variables of the definition. Specifically, 
$(P,I,t) \models ( \langle q_1, \dots, q_m \rangle \mapsto \mu)$ with 
$\mu = \langle b_1, \dots, b_m \rangle$ means that 
$(P,I,t) \models q_i$ if $b_i = 1$ and $(P,I,t) \not\models q_i$ otherwise.
The assignment determines the transformation $\tau = \phi(\mu)$,
which in turn determines the next element $x = \tau(x')$ from the
previous one $x'$.
Point~6 defines the semantics of transformation definitions. It specifies when
variables $p_1, \dots, p_n$ defined by a transformation definition $d$ are true.
The semantics is defined considering a single variable $p_i$ at a time,
considered as part of the former list.
There is an element $x$ of the transformation domain of 
$\mathcal{T} = \langle X, T, \phi, \psi \rangle$ that is
currently associated with the definition $d$.
Namely, the condition $(P,I,t) \models (d \mapsto x)$ holds.
Hence, the assignment $\langle b_1, \dots, b_n \rangle$ to the variables 
$p_1, \dots, p_n$ is given by $\psi(x)$.

\subsection{Examples of Operators}
\label{sec:example-operators}

The mechanism to define transformation operators allows for a great variety of
operators.
In this section we present several examples of transformation operators, showing
that one can easily capture existing operators from the literature or design
novel operators to capture known patterns on sequences.

\paragraph{Temporal Operators.}
We can define operators in the style of the temporal operators from
Past LTL, cf.\ \cite{manna1991completing}. For instance, we can define the
operator
\begin{align*}
  \peventuallyop = \langle \mathbb{B}, T, \phi, \mathit{id} \rangle,
\end{align*}
where $T$ consists of the transformations $\mathit{set}(x) = 1$ and
$\mathit{id}$, and the function $\phi$ is defined as $\phi(0) = \mathit{id}$ 
and $\phi(1) = \mathit{set}$.
Then we can write a definition 
\begin{align*}
  p \ruleif \peventuallyop(a \mid 0),
\end{align*}
which defines $p$ as true when $a$ has happened. Here the only meaningful choice
of the initial transformation element is $0$, and hence it can be omitted.
Thus we can write the same definition as
\begin{align*}
  p \ruleif \peventuallyop a,
\end{align*}
with the understanding that it corresponds to the one above.
Other temporal operators can be introduced in a similar way. 

\paragraph{Threshold Counter Operators.}
The \emph{threshold counter operator} with threshold value $n$ is 
\begin{align*}
  \mathcal{T}_n = \langle \mathbb{N}, T, \phi, \psi \rangle,
\end{align*}
where $T$ consists of the transformations $\mathit{inc}$ and $\mathit{id}$, 
with $\mathit{inc}(x) = x+1$;
the function $\phi$ is defined as
$\phi(1) = \mathit{inc}$ and $\phi(0) = \mathit{id}$; and 
the function $\psi$ is defined as $\psi(x) = 1$ if $x \geq n$ and $\psi(x) = 0$
otherwise.
The operator allows one to check whether a given condition has occurred at least 
$n$ times.
Notably, we can define the operator equivalently as a finite operator by
replacing the set of all natural numbers $\mathbb{N}$ with the finite set
$[0,n]$ and modifying the increment transformation as
$\mathit{inc}(x) = \min(n, x+1)$.
They are equivalent because $\psi$ will not distinguish integers greater
than $n$.

\begin{example}
  An agent must collect at least 30 units of stone, and at least
  115 units of iron given that 13 have already been collected. 
  When it has collected a sufficient number of units, it can
  deliver and get rewarded.
  The reward function is described by the query $(P, \mathit{reward})$ where $P$
  consists of the following definitions:
  \begin{align*}
    \mathit{enoughStone} \ruleif & \mathcal{T}_{30}(\mathit{stone} \mid 0),
    \\
    \mathit{enoughIron} \ruleif & \mathcal{T}_{115}(\mathit{iron} \mid 13),
    \\
    \mathit{successfulDelivery} \ruleif & \mathit{enoughStone} \land
    \mathit{enoughIron} 
    \\
    & \land \mathit{delivery},
    \\
    \mathit{alreadyDelivered} \ruleif & \peventuallyop
    \mathit{successfulDelivery},
    \\
    \mathit{notAlreadyDelivered} \ruleif & \neg \mathit{alreadyDelivered},
    \\
    \mathit{reward} \ruleif & \mathit{delivery} \land
    \mathit{notAlreadyDelivered}.
  \end{align*}
  The proposition $\mathit{reward}$ holds true at the first time point when the
  agent succeeds in a delivery.
\end{example}

\paragraph{Parity Operator.}
The \emph{parity operator} is 
\begin{align*}
  \mathcal{P} = \langle \mathbb{B}, T, \phi, \mathit{id} \rangle,
\end{align*}
where $T$ consists of the identity function $\mathit{id}$ and the Boolean
negation function $\neg$; and
the function $\phi$ is defined as $\phi(0) = \mathit{id}$ and $\phi(1) = \neg$.
The operator checks whether the input variable has been true an even number of
times. It is exemplified by the program 
\begin{align*}
  \mathit{even} & \ruleif \mathcal{P}(a \mid 0),
  \\
  \mathit{odd} & \ruleif \neg \mathit{even}.
\end{align*}
which defines whether $a$ has happened an even or odd number of times.

\paragraph{Metric Temporal Operators.}
In the style of metric temporal logic \cite{koymans1990mtl},
we can define operators such as one that checks whether something happened in
the last $k$ steps---with time isomorphic to the naturals rather than to the
reals as in the original metric temporal logic. The aforementioned operator is
defined as 
\begin{align*}
  \peventuallyop_k = \langle\mathbb{Z}, T, \phi, \psi \rangle,
\end{align*}
where the set $T$ has transformations $\mathit{set}$ and $\mathit{dec}$ defined
as $\mathit{set}(x) = k$ and $\mathit{dec}(x) = x-1$;
the function $\phi$ is defined as 
$\phi(1) = \mathit{set}$ and
$\phi(0) = \mathit{dec}$;
and the function $\psi$ is defined as $\psi(x) = 1$ if $x > 0$ and $\psi(x) = 0$
otherwise. The operator can be equivalently defined as a finite
operator by replacing $\mathbb{Z}$ with $[0,k]$, and definining 
$\mathit{dec}(x) = \max(0,x-1)$.

\paragraph{Operators with infinite transformation domain.}
While all the operators above can be defined as finite operators,
one can also include operators that require an infinite transformation domain.
For instance, the operator
\begin{align*}
  \mathrm{sameNum} = \langle \mathbb{Z}, T, \phi, \psi \rangle,
\end{align*}
where $T$ consists of $\mathit{inc}(x) = x+1$, $\mathit{id}(x) = x$, and 
$\mathit{dec}(x) = x-1$;
the function $\phi$ is defined as
$\phi(0,0) = \phi(1,1) = \mathit{id}$, 
$\phi(1,0) = \mathit{inc}$,
$\phi(0,1) = \mathit{dec}$;
and the function $\psi$ is defined as $\psi(x) = 1$ iff $x = 0$.
When the operator is used in a definition such as
\begin{align*}
  p \ruleif \mathrm{sameNum}(a,b \mid 0),
\end{align*}
we have that variable $p$ is true iff $a$ and $b$ have been true the same number
of times.
The operator can be used to recognise the Dyck language of
balanced parentheses, which is not regular. 
The existence of an equivalent finite operator would imply the existence of a
finite automaton, and hence regularity of the language.

\subsection{Finite Semigrouplike Operators}

We present a principled way to define operators in terms of finite
semigroups.

\begin{definition} \label{definition:semigrouplike_operator}
  Let us consider 
  a finite semigroup $S = (X, {}\cdot{})$, 
  a surjective function $\phi: \mathbb{B}^m \to X$, and 
  an injective function $\psi: X \to \mathbb{B}^n$.
  They define the transformation operator 
  $\langle X, T, \phi, \psi \rangle$ where
  $T$ consists of each transformation $y(x) = x \cdot y$ for $y \in X$.
  We call it a \emph{semigrouplike operator}.
\end{definition}

Intuitively, $\phi$ is a binary decoding of the set $X$, and $\psi$ is a binary
encoding.
Note that $y$ is seen both as an element of $X$ and as a function $y : X \to X$.
From now on we assume that for every set $X$ such an encoding and decoding is
fixed. Any choice will be valid for our purposes.
Thus, we simply say that a semigroup defines a semigrouplike operator,
without mentioning the functions $\phi$ and $\psi$ explicitly.

\begin{definition}[Prime operators]
  \label{definition:prime_operators}
  The \emph{flip-flop operator} is the transformation operator defined by the
  flip-flop monoid.
  An operator is a \emph{(simple) group operator} if it is defined by a (simple)
  finite group.
  The \emph{prime operators} are the flip-flop operator and the simple group
  operators. 
\end{definition}

We will focus in particular on the flip-flop operator and on operators defined
by cyclic groups.
For these two operators we provide an explicit definition, which requires us to
commit to a choice of the encoding/decoding functions.
The definitions we provide are improved with respect to the ones following from
a direct application of Definition~\ref{definition:semigrouplike_operator}. In
particular, we omit one element from the transformation domain of the flip-flop
operator, since it would be redundant.
\ifarxiv
In appendix we discuss how to derive the operators by direct application of 
Definition~\ref{definition:semigrouplike_operator}.
\fi

\begin{definition} \label{definition:flipflop_operator}
  The \emph{flip-flop operator} is
  \begin{align*}
    \mathcal{F} = \langle \mathbb{B}, T, \phi, \mathit{id} \rangle,
  \end{align*}
  where $T$ consists of the transformations $\mathit{set}, \mathit{reset},
  \mathit{read}$ defined as
  \begin{align*}
      \mathit{set}(x) = 1,
    \quad
      \mathit{reset}(x) = 0,
    \quad
      \mathit{read}(x) = x,
  \end{align*}
  and the function $\phi$ is defined as
  \begin{align*}
    \phi(0,0) = \mathit{read},\;
    \phi(1,0) = \phi(1,1) = \mathit{set},\;
    \phi(0,1) = \mathit{reset}.
  \end{align*}
\end{definition}

The flip-flop operator corresponds to the flip-flop from digital circuits---it
corresponds to an \emph{SR latch}, where input $11$
is not allowed though, cf.\ \cite{roth2004fundamentals}.
The operator allows us to specify a flip-flop with a definition as
\begin{align*}
  \mathit{storedBit} \ruleif 
  \mathcal{F}(\mathit{writeOne},\mathit{writeZero}).
\end{align*}
Therefore the logic $\mathcal{L}(\mathcal{F})$ can be employed as
a specification language for digital circuits with logic gates and flip-flops.
Next we introduce the cyclic operators.

\begin{definition} \label{definition:cyclic_operator}
  The \emph{cyclic operator of order $n$} is
  \begin{align*}
    \mathcal{C}_n = \langle [0,n-1], T, \phi, \psi \rangle,
  \end{align*}
  where the transformations are 
  $T = \{ \mathit{inc}_i \mid i \in [0,n-1] \}$ defined as
  \begin{align*}
    \mathit{inc}_i(x) = (x + i)\, \operatorname{mod} n,
  \end{align*}
  the function $\phi$ is defined as
  \begin{align*}
  \phi(b_1, \dots, b_m) = \mathit{inc}_i,
  \end{align*}
  where $m$ is the minimum number of bits required to represent $n$,
  and $i$ is the minimum between $n-1$ and the number whose binary
  representation is $b_1 \dots b_m$;
  finally, the function $\psi(x)$ yields the binary representation of $x$.
\end{definition}

When a cyclic operator is used in a definition as
\begin{align*}
  p_1, \dots, p_m \ruleif \mathcal{C}_n(0, \dots, 0, a \mid 0),
\end{align*}
variables $p_1, \dots, p_m$ provide the binary representation of the number of
times $a$ has been true modulo $n$.
Note that the parity operator introduced earlier coincides with
$\mathcal{C}_2$.
\ifarxiv
Also note that we have chosen an efficient binary encoding/decoding scheme.
An alternative could be to choose a function $\psi$ that yields the
one-hot encoding of the counter value. This would require $m=n$, instead of $m$
logarithmic in $n$.
\fi
Next we characterise the prime cyclic operators.

\begin{proposition}
  A cyclic operator $\mathcal{C}_n$ is a prime operator iff $n$ is a prime
  number.
\end{proposition}

The other prime operators are defined by finite simple groups as
found in \emph{the classification of finite simple groups}, cf.\
\cite{gorenstein2018classification}. In addition to the infinite family of
cyclic groups of prime order, the classification also includes the two infinite
families of \emph{alternating groups} and groups of \emph{Lie type}, along with 
the 27 \emph{sporadic} groups.

\section{Expressivity Results}

We show expressivity results for Transformation Logics featuring
finite operators, with a focus on semigroup-like operators.
We start by defining the notion of expressivity for a Transformation
Logic, in terms of formal languages.

\begin{definition}
  Consider a program $P$ with input variables 
  $V = \{ p_1, \dots, p_n \}$.
  Every letter $\langle b_1, \dots, b_n \rangle$ of the alphabet 
  $\mathbb{B}^n$ defines an assignment $\nu$ to the variables in $V$, as
  $\nu(p_i) = b_i$.
  Then, every word $w$ over $\mathbb{B}^n$ defines an input $I_w$ to $P$.
  A query $(P,q)$ \emph{accepts} a word $w$ over $\mathbb{B}^n$ iff 
  $(P, I_w, |w|) \models q$;
  and it \emph{recognises} a language $L$ over $\mathbb{B}^n$ if it accepts
  exactly the words in $L$.  
\end{definition}

\begin{definition}
  The \emph{expressivity} of a logic $\mathcal{L}$ is the set of
  languages recognised by its queries.
  A logic $\mathcal{L}_1$ is \emph{less expressive} than a logic
  $\mathcal{L}_2$, written $\mathcal{L}_1 \sqsubset \mathcal{L}_2$, if the
  expressivity of $\mathcal{L}_1$ is properly included in the expressivity of
  $\mathcal{L}_2$.
\end{definition}

First, we establish the expressivity when all finite operators are included.

\begin{restatable}{theorem}{theoremexpressivitytransformationlogics}
  \label{theorem:expressivity_transformation_logics}
  For $\mathbf{A}$ the set of all finite operators,
  the expressivity of 
$\mathcal{L}(\mathbf{A})$ is the regular languages.
\end{restatable}

The result follows from the fact that finite operators capture finite automata;
conversely, programs with finite operators can be mapped to a composition of
finite automata. In particular to a so-called \emph{cascade composition} of
automata. This characterisation allows us to employ 
\emph{algebraic automata theory}, cf.\
\cite{ginzburg,arbib1969theories,domosi2005algebraic}, in proving that prime
operators suffice to capture regular languages.

\begin{restatable}{theorem}{theoremexpressivitytransformationlogicsprime}
  \label{theorem:expressivity_transformation_logics_prime}
  For $\mathbf{P}$ the set of all prime operators,
  the expressivity of the logic $\mathcal{L}(\mathbf{P})$ is the regular
  languages.
\end{restatable}

The result is obtained by showing that every transformation definition with a
finite transformation operator can be captured by a program with prime operators
only.
The prime operators to use are suggested by the \emph{prime decomposition
theorem} for finite automata \cite{krohn1965rhodes}.
For star-free regular languages, we obtain the following specialised result.

\begin{restatable}{theorem}{theoremexpressivitykrlogicsstarfree}
  \label{theorem:expressivity_kr_logics_starfree}
  For $\mathcal{F}$ the flip-flop operator,
  the expressivity of the logic $\mathcal{L}(\mathcal{F})$ is the star-free
  regular languages.
\end{restatable}

Beyond star-free, we have nameless fragments of the regular languages. 
Here we start their exploration. First, to go beyond star-free it suffices to
introduce group operators.

\begin{restatable}{theorem}{theoremgroupsaddexpressivity}
  \label{theorem:groups_add_expressivity}
  For any non-empty set $\mathbf{G}$ of group operators,
  the logic $\mathcal{L}(\mathcal{F}, \mathbf{G})$ is strictly
  more expressive than $\mathcal{L}(\mathcal{F})$.
\end{restatable}

Next we focus on cyclic operators. We show them to yield a core of the
Transformation Logics with favourable properties.
First, cyclic operators along with the flip-flop operator 
form a canonical and universal set of operators for the logic 
$\mathcal{L}(\mathcal{F},\mathbf{S})$, where $\mathbf{S}$ is the set of all
solvable group operators.
They are canonical and universal for $\mathcal{L}(\mathcal{F},\mathbf{S})$ in
the same way the operators $\{ \land, \neg \}$ are canonical and
universal for propositional logic.
We state universality and then canonicity.

\begin{restatable}[Universality]{theorem}{theoremexpressivityofcyclicoperators}
  \label{theorem:expressivity_of_cyclic_operators}
  For $\mathbf{C}$ the set of all cyclic operators of prime order, 
  and $\mathbf{S}$ any set of solvable operators, 
  the expressivity of $\mathcal{L}(\mathcal{F}, \mathbf{C})$ includes the
  expressivity of $\mathcal{L}(\mathcal{F}, \mathbf{S})$.
\end{restatable}

\begin{restatable}[Canonicity]{theorem}{theoremexpressivitycontainment}
  \label{theorem:expressivity_containment}
  Given two sets $\mathbf{C}_1$ and $\mathbf{C}_2$ of cyclic
  operators of prime order, the logics 
  $\mathcal{L}(\mathcal{F}, \mathbf{C}_1)$ and $\mathcal{L}(\mathcal{F},
  \mathbf{C}_2)$ have the same expressivity if and only if\/ 
  $\mathbf{C}_1 = \mathbf{C}_2$.
\end{restatable}

The canonicity result implies the existence of infinite expressivity hierarchies
such as the following one, if we note that every $\mathcal{C}_p$ is a prime
operator when $p$ is a prime number.
\begin{corollary}
  $\mathcal{L}(\mathcal{C}_2) \sqsubset \mathcal{L}(\mathcal{C}_2,
  \mathcal{C}_3) \sqsubset \mathcal{L}(\mathcal{C}_2, \mathcal{C}_3,
  \mathcal{C}_5) \sqsubset \cdots$.
\end{corollary}

It is worth noting that this form of canonicity does not hold for prime
operators in general.
\begin{restatable}{theorem}{theoremnoncanonicity}
  \label{theorem:non_canonicity}
  There are two sets $\mathbf{P}_1 \subset \mathbf{P}_2$ of prime
  operators such that
      $\mathcal{L}(\mathbf{P}_1)$ and 
      $\mathcal{L}(\mathbf{P}_2)$ have the same expressivity.
\end{restatable}

Combining our expressivity results together, we obtain the following hierarchy
theorem.
\begin{restatable}{theorem}{theoremhierarchy}
  \label{theorem:hierarchy}
  For $\mathbf{S}$ all solvable group operators,
  and $\mathbf{P}$ all prime operators,
  the following infinite hierarchy of expressivity holds:
  \begin{align*}
    \mathcal{L}(\mathcal{F}) 
    \!\sqsubset\!
    \mathcal{L}(\mathcal{F}, \mathcal{C}_2) 
    \!\sqsubset\!
    \mathcal{L}(\mathcal{F}, \mathcal{C}_2, \mathcal{C}_3) 
\!\sqsubset\!
    \cdots 
    \!\sqsubset\!
    \mathcal{L}(\mathcal{F}, \mathbf{S})
    \!\sqsubset\!
    \mathcal{L}(\mathbf{P}).
  \end{align*}
\end{restatable}

The theorem provides our current picture of the expressivity of the
Transformation Logics.
At the bottom of the hierarchy we have 
$\mathcal{L}(\mathcal{F})$ with the expressivity of the star-free regular
languages.
At the top of the hierarchy we have 
$\mathcal{L}(\mathbf{P})$ with the expressivity of all regular languages.
Between them we have infinitely-many logics capturing distinct fragments of the
regular languages.

\section{Complexity Results}

We study the complexity of the evaluation problem for the Transformation Logics.
\begin{definition}
  The \emph{evaluation problem} of a Transformation Logic $\mathcal{L}$ is the
  problem to decide, 
  given a query $(P,q)$ with $P \in \mathcal{L}$, and an input 
  $I = I_1, \dots, I_\ell$\/ for $P$, whether it holds that 
  $(P,I,\ell) \models q$.
\end{definition}

To study the complexity we need to assume a representation for the operators,
which determines the \emph{size of an operator}. The choice of a representation
for operators in general is beyond the scope of this paper. We discuss two
concrete cases below.

\begin{definition} \label{definition:polytime_operator}
  A family $\mathbf{T}$ of operators is \emph{polytime} if, for every operator 
  $\langle X, T, \phi, \psi \rangle \in \mathbf{T}$ with
  $\phi: \mathbb{B}^m \to T$ and $\psi: X \to \mathbb{B}^n$, it holds that, 
  for every $x \in X$ and every $\mu \in \mathbb{B}^m$,
  the value $\psi(\tau(x))$ with $\tau = \phi(\mu)$ can be computed in time
  polynomial in the size of the operator.
\end{definition}

\begin{restatable}{theorem}{theoremcomplexity}
  \label{theorem:complexity}
  For any (possibly infinite) set $\mathbf{T}$ of polytime operators,
  the evaluation problem of $\mathcal{L}(\mathbf{T})$ is in $\textsc{PTime}$.
  Furthermore, there exists a set $\mathbf{H}$ of polytime operators such that
  the evaluation problem of $\mathcal{L}(\mathbf{H})$ is
  $\textsc{PTime}$-complete.
\end{restatable}

We argue the upper bound applies to every finite set of finite operators,
including prime operators. On the contrary, infinite sets of finite
operators may not be polytime.

\begin{restatable}{lemma}{lemmacomplexityfiniteoperators}
  \label{lemma:complexity_finite_operators}
  Every finite set of finite operators is polytime.
  There exists a set of finite operators that is not polytime.
\end{restatable}

\begin{restatable}{theorem}{theoremcomplexityfiniteoperators}
  \label{theorem:complexity_finite_operators}
  For any finite set $\mathbf{T}$ of finite operators,
  the evaluation problem of $\mathcal{L}(\mathbf{T})$ is in $\textsc{PTime}$.
\end{restatable}

Next we focus on threshold counter operators $\mathcal{T}_n$ and cyclic
operators $\mathcal{C}_n$. We denote the complete families as 
$\{ \mathcal{T}_n \}$ and $\{ \mathcal{C}_n \}$.
Notably, such operators can be represented compactly. 

\begin{definition}
  A representation for 
  the family of threshold counter operators $\{ \mathcal{T}_n \}$ is said to be 
  \emph{compact} if the size of the representation of $\mathcal{T}_n$ is 
  $O(\log n)$.
  Similarly for the family of cyclic operators $\{ \mathcal{C}_n \}$.
\end{definition}

A compact representation can be obtained by encoding the symbols
$\mathcal{T}$ and $\mathcal{C}$ with a constant number of bits, and the index
$n$ in binary with a logarithmic number of bits.

\begin{restatable}{proposition}{propositioncompactoperators}
  \label{proposition:compact_operators}
  The threshold-counter operators $\{ \mathcal{T}_n \}$ and the  
  cyclic operators $\{ \mathcal{C}_n \}$
  admit a compact representation.
\end{restatable}

\begin{restatable}{lemma}{lemmacomplexitycounters}
  \label{lemma:complexity_counters}
  The families of operators $\{ \mathcal{T}_n \}$ and $\{ \mathcal{C}_n \}$ are
  polytime, even when represented compactly.
\end{restatable}

Overall we obtain that polynomial time evaluation is possible for every
Transformation Logic that includes any finite number of finite operators, along
with the threshold-counter and cyclic operators.

\begin{restatable}{theorem}{theoremoverallcomplexity}
  \label{theorem:overall_complexity}
  For any finite set $\mathbf{A}$ of finite operators,
  the evaluation problem of 
  $\mathcal{L}(\mathbf{A}, \{ \mathcal{T}_n \}, \{ \mathcal{C}_n \})$
  is in $\textsc{PTime}$ even under compact representation of 
  the operators $\{ \mathcal{T}_n \}$ and $\{ \mathcal{C}_n \}$.
\end{restatable}

\subsection{Constant-Depth and Data Complexity}

We study the complexity of evaluation of programs having a constant depth.
It implies \emph{data complexity} results. 
Specifically, all membership complexity results for constant-depth programs
imply also membership in data complexity, when the program is fixed and the
size of the input to the program is arbitrary.
Data complexity measures how evaluation of a fixed program scales with the size
of the input \cite{vardi1982complexity}.
We first define the depth of a program, and corresponding classes of
constant-depth programs.

\begin{definition}
  Consider a program $P$.
  The \emph{depth} of an input variable of $P$ is zero.
  The \emph{depth} of a variable defined by a static definition $d \in P$ is the
  maximum
  depth of a variable in the body of $d$ plus the depth of the parse-tree of
  the body of $d$.
  The \emph{depth} of a variable defined by a delay or transformation definition 
  $d \in P$ is the maximum depth of a variable in the body of $d$ plus one.
  The \emph{depth of $P$} is the maximum depth of a variable in $P$.
\end{definition}

\begin{definition}
  For any set of transformation operators $\mathbf{T}$, and any depth $k$,
  the Transformation Logic $\mathcal{L}(\mathbf{T} \,|\, k)$ is the subset of 
  $\mathcal{L}(\mathbf{T})$ with programs of depth at most $k$.
\end{definition}

Our results are phrased in terms of three circuit complexity classes,
reported below with the known inclusions.
\begin{align*}
  \textsc{AC}^0 \subsetneq \textsc{ACC}^0 \subseteq \textsc{NC}^1
\end{align*}

The first result is for the flip-flop operator, and it builds on
a result for the complexity of group-free semigroups
\cite{chandra1985unbounded}.
\begin{restatable}{theorem}{theoremconstantdepthflipflop}
  \label{theorem:constant_depth_flipflop}
  For any $k$,
  evaluation of $\mathcal{L}(\mathcal{F} \,|\, k)$ is in
  $\textsc{AC}^0$.
\end{restatable}

Cyclic group operators, and solvable group operators in general, increase the
complexity of the evaluation problem.

\begin{restatable}{theorem}{theoremconstantdepthacc}
  \label{theorem:constant_depth_acc}
  For any depth $k$, and 
  any finite set $\mathbf{S}$ of solvable group operators, 
  evaluation of
  $\mathcal{L}(\mathcal{F}, \mathbf{S} \,|\, k)$ is in $\textsc{ACC}^0$.
  Furhermore, there is a solvable group operator $\mathcal{G}$ such that,
  for every depth $k \geq 1$,
  evaluation of $\mathcal{L}(\mathcal{G} \,|\, k)$ is not in 
  $\textsc{AC}^0$.
\end{restatable}

The upper bound makes use of a result for the complexity of solvable
groups \cite{barrington1989bounded}.
The lower bound relies on the fact that the cyclic (hence solvable) group $C_2$
captures the parity function, known not to be in
$\textsc{AC}^0$ \cite{furst1984parity}.

Non-solvable groups introduce an exact correspondence with 
the larger cicuit complexity class $\textsc{NC}^1$.
Our result builds on a result for the complexity of
non-solvable groups \cite{barrington1989bounded}.

\begin{restatable}{theorem}{theoremcomplexitync}
  \label{theorem:complexity_nc}
  For any depth $k$, and
  any finite set $\mathbf{G}$ of groups containing a non-solvable group,
  the evaluation problem of
  $\mathcal{L}(\mathcal{F},\mathbf{G} \,|\, k)$ is complete for $\textsc{NC}^1$
  under $\textsc{AC}^0$ reductions.
\end{restatable}

\section{Relationship with Past LTL}

We show the Transformation Logic $\mathcal{L}(\mathcal{F})$ captures Past LTL. 
\ifarxiv
Background on Past LTL is included in the appendix.
\fi
First, the \emph{before} and \emph{since} operators correspond to the delay and
flip-flop operators, respectively.

\begin{restatable}{lemma}{lemmapltlbeforeop}
  \label{lemma:pltl_beforeop}
  Consider a Past LTL formula $\varphi = \beforeop p$. Let $P$ be the singleton program consisting of the definition
$q \ruleif \operatorname{\mathcal{D}} p$.
For every interpretation $I$ of $\varphi$ and every time point $t$,
  it holds that $(I, t) \models \varphi$ iff 
  $(P, I, t) \models q$.
\end{restatable}

\begin{restatable}{lemma}{lemmapltlsinceop}
  \label{lemma:pltl_sinceop}
  Consider a Past LTL formula $\varphi = a \sinceop b$.
Let $P$ be the program consisting of the two definitions
$c \ruleif \neg a$ and
$p \ruleif \mathcal{F}(b, c \,|\, 0)$.
For every interpretation $I$ of $\varphi$ and every time point $t$,
  it holds that $(I, t) \models \varphi$ iff 
  $(P, I, t) \models p$.
\end{restatable}

Given the lemmas above, we can build a program for any given Past LTL formula by
induction on its parse-tree, introducing one static definition for each
occurence of a Boolean operator, one delay definition for each occurrence of the
before operator, and one transformation definition for each occurrence of the
since operator.
\begin{restatable}{theorem}{theorempltltranslation}
  \label{theorem:pltl_translation}
  Every Past LTL formula $\varphi$ can be translated into a program  
  of $\mathcal{L}(\mathcal{F})$.
Such a program has size linear in the size of $\varphi$, and 
  it can computed from $\varphi$ in logarithmic space.
\end{restatable}

\section{Related Work}
\label{sec:related-work}

The Transformation Logics are a valuable addition to the rich set of temporal
and dynamic logics adopted in AI.
Such logics include \emph{propositional} temporal logics of the past such as Past LTL,
cf.\ \cite{manna1991completing};
temporal logics of the future interpreted both on finite and infinite traces 
\cite{pnueli1977temporal,wolper1983temporal,degiacomo2013ltlf};
dynamic logics such as PDL, cf.\
\cite{harel2000dynamic}, and linear dynamic logic on finite traces
\cite{degiacomo2013ltlf}; 
logics with a dense interpretation of time such as
metric temporal logic \cite{koymans1990mtl}.
They also include \emph{first-order} variants such as Past FOTL 
\cite{chomicki1995efficient}.
Finally, they include \emph{rule-based} logics of several kinds:
propositional modal temporal logics such as Templog
\cite{abadi1989temporal,baudinet1995expressive};
logics with first-order variables where time occurs explicitly as an
argument of a special sort such as  Datalog$_\mathrm{1S}$ \cite{chomicki1988temporal},
and Temporal Datalog \cite{ronca2018stream,ronca2022delay};
logics with metric temporal operators such as
DatalogMTL \cite{brandt2018querying,walega2019datalogmtl,walega2020datalogmtl};
and propositional modal logics specialised for reasoning on streams
\cite{beck2018lars}, and weighted variants thereof \cite{eiter2020weighted}
which share an algebraic flavour with our work.

\section{Conclusions}

The Transformation Logics provide a general framework for logics over sequences.
The flexibility given by the transformation operators allows for defining logics
with the expressivity of many different fragments of the regular languages, 
under the guidance of group theory.
While the star-free regular languages are well-known as they are the
expressivity of many well-established formalisms, the fragments beyond them are
less known.
The Transformation Logics provide a way to explore them.

\section*{Acknowledgments}
Alessandro Ronca is supported by the European Research Council (ERC) under the
European Union's Horizon 2020 research and innovation programme (Grant
agreement No.\ 852769, ARiAT).

\bibliographystyle{named}
\bibliography{bibliography}

\ifarxiv
 \onecolumn
 \newpage
 \appendix

\section*{\LARGE Appendix of `The Transformation Logics'}
\medskip

\paragraph{Summary.}
The appendix consists of six parts.
\begin{description}
  \item
    \emph{Part~\ref{sec:motivation-example}.} We provide the full version of
    Example~1 and Example~2 from the introduction.
  \item
    \emph{Part~\ref{appendix:extending-past-ltl}.}
    We illustrate how Past LTL could be extended using transformation
    operators.
  \item
    \emph{Part~\ref{appendix:semigrouplike-operators}.}
    We explain how the definitions of flip-flop and cyclic operators are
    derived.
  \item
    \emph{Part~\ref{sec:additional-background}.}
    We introduce additional background required by the proofs.
  \item
    \emph{Part~\ref{sec:appendix-preliminaries}.}
    We prove some preliminary results used in the proofs of the main results.
  \item
    \emph{Part~\ref{sec:proofs-expressivity}.}
    We present the proofs of all our results, following the order in which the
    results are stated in the main body. 
  \end{description}

\section{Full Version of the Examples}
\label{sec:motivation-example}

\subsection{Example~1}

We provide full details on Example~1 reported in the introduction section.
It is an example in an agent-based setting, which can be cast into
Reinforcement Learning applications as well as Planning applications.

Let us consider an agent, that can accomplish a task $\mathit{assignedTask}$
multiple times.
We would like to know whether the agent has accomplished the task on every past
day, considering that we get an update every minute.
Namely, we receive as input a sequence $I = I_1, \dots, I_\ell$ where
$\mathit{assignedTask} \in I_t$ iff $t$ is an instant at which
the agent has just completed the assigned task.
The specification can be captured in a Transformation Logic that features an
operator which signals the end of every day. 

The operarator is defined as follows, letting $n = 24*60 = 1440$ be the
minutes in a day.
\begin{align*}
  \operatorname{endOfDay} = \langle [0,n-1], T, \phi, \psi \rangle
\end{align*}
We have that the set of transformations is $T = \{ \mathit{inc} \}$ with 
$\mathit{inc}(x) = x+1 \operatorname{mod} n$,
the function $\phi: \mathbb{B}^0 \to T$ is the constant map to $\mathit{inc}$,
and the function $\psi$ is defined as $\psi(b_0,\dots, b_{10}) = 1$ iff 
$b_0, \dots, b_{10}$ is the binary encoding of $n-1$---note that eleven bits
suffice to encode $n$ numbers.

The example is captured by the following program.
\begin{align}
  \label{rule1}
  \mathit{end} & \ruleif \operatorname{endOfDay}(\mid 0)
\\
  \label{rule2}
  \mathit{doneTaskToday} & \ruleif \mathcal{F}(\mathit{assignedTask},
  \mathit{end} \mid 0)
  \\
  \label{rule3}
  \mathit{failedInADay} & \ruleif \mathit{end} \land \neg \mathit{doneTaskToday}
  \\
  \label{rule4}
  \mathit{failed} & \ruleif \mathcal{F}(\mathit{failedInADay}, \bot \mid 0)
\end{align}
First, we have that \eqref{rule1} defines $\mathit{end}$ to be true at
the end of every day.
Then, \eqref{rule2} defines $\mathit{doneTaskToday}$ to be true if the task
has been completed since the end of the previous day---note that the flip-flop
sets when the task is completed, but it resets at the end of every day.
Then, \eqref{rule3} establishes that the agent has failed to accomplish the task
in the current day if it is the end of the day and 
$\mathit{doneTaskToday}$ is false.
In such a case, \eqref{rule4} raises the flag $\mathit{failed}$,
permanently.

The characteristic semigroup of $\operatorname{endOfDay}$ is the cyclic group
$C_{1440}$, and hence it can be expressed using the prime operators 
$\mathcal{C}_2$, $\mathcal{C}_3$, $\mathcal{C}_5$, by 
Theorem~\ref{theorem:prime_decomposition_of_transformation_logics} proven below.
This would allow us to get rid of the end-of-day operator, in favour of prime
operators.

The program belongs to a Transformation Logic with a data complexity in
$\textsc{ACC}^0$ according to Theorem~\ref{theorem:constant_depth_acc}, since
the program uses the flip-flop operator and solvable group operators.

This example cannot be captured in Past LTL, since 
the language recognised by the query $(P,\mathit{failed})$ is not star-free.
This can be seen by verifying that the language is not \emph{noncounting}, cf.\
\cite{ginzburg}.

\subsection{Example~2}

We provide full details on Example~2 reported in the introduction section.

A cycling race with $n$ participants takes place, and we need to keep track of
the live ranking.
At each step an overtake can happen, in which case it is communicated to us in
the form $(i,i+1)$ meaning that the cyclist in position $i$ has overtaken the
one in position $i+1$. Specifically, we receive a proposition
$\mathit{overtake}(i,i+1)$.
We know the initial ranking, and we need to keep track of the live ranking.

The ranking corresponds to the symmetric group $S_n$.
In fact, the symmetric group $S_n$ is the group of all permutations of
a list of $n$ objects. It can be seen as the set of all transformations over the
integer interval $[1,n]$. 
For example, for $n = 3$, the transformation $\{ 1 \mapsto 2, 2 \mapsto 3, 3
\mapsto 1 \}$ describes the permutation where we move the first element to the
second position, the second element to the third position, and the third element
to the first position.
The symmetric group is generated by the set of \emph{adjacent transpositions}.
An \emph{adjacent transposition} is a permutation that swaps two adjacent
elements and does not affect any other element.
For example, for $n = 3$, the transformation $\{ 1 \mapsto 2, 2 \mapsto 1, 3
\mapsto 3 \}$ is an adjacent transposition. 
We observe that the ranking of the example is a list of cyclicists, and the 
`overtakes' are adjacent transpositions.
Thus the example amounts to the symmetric group $S_n$.

The transformation operator defined by $S_n$ has a large number 
of transformations; this number is $n!$ since it is the number of permutations.
However we are only interested in adjacent transpositions, along with the
identity permutation to capture the case when no overtake happens. This gives us
a total number of $n$ transformations.
Thus we consider a simplified version of the operator defined by $S_n$, having 
only transformations given by adjacent transpositions, along with identity.
The operarator is defined as follows.
\begin{align*}
  \operatorname{ranking} = \langle X, T, \phi, \psi \rangle
\end{align*}
The domain $X$ is the set of lists of the integers $[1,n]$.
The set of transformations is 
$T = \{ \mathit{id} \} \cup 
\{ \mathit{transpose}_{(i,i+1)} \mid i \in [1,n-1] \}$ 
where 
\begin{align*}
  \mathit{transpose}_{(i,i+1)}(j_1, \dots, j_n) = j_1, \dots, j_{i-1}, j_{i+1},
  j_i, j_{i+2}, \dots, j_n.
\end{align*}
The function $\phi$ maps the all-zero vector
$\langle 0, \dots, 0 \rangle$ to $\mathit{id}$, and it maps
every indicator vector $\langle b_1, \dots, b_m \rangle$ indicating a number
$i$ to the transformation $\mathit{transpose}_{(i,i+1)}$. For example, the
indicator vector $\langle 0, 0, 1, 0, 0 \rangle$ is mapped to 
$\mathit{transpose}_{(3,4)}$.
Let $m$ be the minimum number of bits to encode $n$.
The function $\psi: X \to \mathbb{B}^{nm}$ maps every list in $X$ to its binary
encoding.

The example is captured by the following singleton program.
\begin{align*}
  p_{1,1}, \dots, p_{1,m}, \dots, p_{n,1}, \dots, p_{n,m} & \ruleif 
  \operatorname{ranking}\big(\mathit{overtake}(1,2), \dots,
  \mathit{overtake}(n-1,n) \mid (c_1, \dots, c_n) \big)
\end{align*}
The list $(c_1, \dots, c_n)$ is the initial ranking.
The head variables specify the live ranking, based on the updates
provided as input.

\section{Extending Past LTL}
\label{appendix:extending-past-ltl}

We first introduce the syntax and semantics of Past LTL, which is required in
this section as well as in the proofs later on. Then we briefly discuss an
aspect regarding the formalisation of the semantics of Past LTL, that becomes
particularly interesting in light of the semantics of the Transformation Logics.
We conclude with the main content of this section, which is a discussion on
possible extensions of Past LTL by means of transformation operators.

\subsection{Background: Past LTL}

We introduce \emph{Past Linear Temporal Logic} (Past LTL), c.f.\
\cite{manna1991completing}.

\medskip
\par
\noindent
\textbf{Syntax.}
A \emph{Past LTL formula} is built out of propositional logic formulas by means
of the \emph{temporal operators} $\{ \beforeop, \sinceop \}$, called the
\emph{before operator} and the \emph{since operator}, respectively.
It is defined inductively as a propositional logic formula or
one of the following two expressions: $\prevop \alpha$,
$\alpha \sinceop \beta$ where $\alpha$ and $\beta$ are Past LTL formulas.
Additional temporal operators can be defined in terms of the previous operators,
e.g., \emph{sometimes in the past}
$\peventuallyop \alpha := \top \sinceop \alpha$ and
\emph{always in the past}
$\palwaysop \alpha :=  \neg \peventuallyop \neg \alpha$.

\medskip
\par
\noindent
\textbf{Semantics.}
An \emph{interpretation for a Past LTL formula $\alpha$} is a finite non-empty
sequence of sets of propositional variables---note that the non-emptiness
requirement does not exclude the sequence consisting of the empty set.
Intuitively, the fact that a variable appears in the $t$-th element of the
intepretation means that the variable stands for a \emph{proposition} that is
true at time $t$.
Note that an interpretation can also be seen as a sequence of assignments.
The semantics of formulas is defined in terms of the following satisfiability
relation, which extends the one for propositional logic formulas, also making it
relative to time $t \in \{ 1, 2, \dots, \}$.
Given a Past LTL formula $\alpha$, and an interpretation $I = I_1, \dots, I_n$
for $\alpha$, the \emph{satisfiability relation} $I,t \models \alpha$ is defined
following the structural definition of formulas:
\begin{itemize}
\item
    $I,t \models a$ iff $a \in I_t$,
  \item
    $I,t \models \neg \alpha$ iff $I,t \not\models \alpha$,
  \item
    $I,t \models \alpha \lor \beta$ iff $I,t \models \alpha$ or
    $I,t \models \beta$,
  \item
    $I,t \models \alpha \land \beta$ iff $I,t \models \alpha$ and
    $I,t \models \beta$.
  \item
    $I,t \models \LTLcircleminus \alpha$ iff
    $I,t-1 \models \alpha$,
  \item
    $I,t \models \alpha \sinceop \beta$ iff
    there exists $j \in [1,t]$ such that $I,j \models \beta$ and, for every
    $k \in [j+1,t]$, it holds that
    $I,k \models \alpha$.
\end{itemize}

The since operator admits an equivalent, inductive,
caracterisation of its semantics, cf.\ \cite{chomicki1995efficient}.
We have that
$I, t \models \alpha \sinceop \beta$ holds
iff any of the two following conditions hold:
\begin{enumerate}
  \item
    $I, t \models \beta$,
  \item
    $I, t-1 \models \alpha \sinceop \beta$ and $I, t \models \alpha$.
\end{enumerate}

\subsection{Discussion of the Semantics of Past LTL}
It is worth noting that the semantics of the since operator is
defined in terms of the logic $\textsc{FOL}[<]$ first-order logic on a
finite linear order. 
The expressivity of $\textsc{FOL}[<]$ is the star-free regular languages.
Thus, it is not possible to use $\textsc{FOL}[<]$ for defining the semantics
of operators whose expressivity is beyond the star-free regular languages.
This justifies the alternative approach adopted for transformation
operators.

\subsection{Extending Past LTL}

We discuss how to extend Past LTL with transformation operators. The aim is
to increase its expressivity beyond star-free, and to make it more
succinct with operators such as threshold counter operators.
Every transformation operator $\mathcal{T}$ with unary output can
be naturally included as an operator in Past LTL.
Thus, the set of formulas can be extended with formulas of the form 
$\mathcal{T}(\varphi_1, \dots, \varphi_n \mid x)$ where $\varphi_i$ is a
formula, and $x$ belongs to the transformation domain of $\mathcal{T}$. 
This yields the logic $\textsc{PLTL}(\mathbf{T})$, the extension of Past LTL
with operators from $\mathbf{T}$.

\begin{example}
  Consider the parity operator $\mathcal{P}$. Then, the formula
  $\palwaysop (a \lor \mathcal{P}(\top \mid 0))$ is true whenever $a$ has always
  happened in odd positions. Note that the disjunction is always satisfied in
  even positions, where $\mathcal{P}(\top \mid 0)$ evaluates to true.
\end{example}

Transformation operators with output arity greater than one do not seem to admit
a natural way to be included in the syntax of Past LTL formulas.
This is not limiting in terms of expressivity as long as one is concerned with
finite operators.
In fact, any finite operator with a transformation domain 
$\{ x_1, \dots, x_n \}$ can be captured by $n$ unary-output operators, each
indicating whether $x_i$ is the current transformation element.
Still, the \wrule{}-based syntax of the Transformation Logics plays in favour
of succinctness and intelligibility.

\section{Flip-Flop and Cyclic Operators: How
Definitions~\ref{definition:flipflop_operator} and
\ref{definition:cyclic_operator} Are Obtained}

\label{appendix:semigrouplike-operators}

We explain how the definition of flip-flop operator
(Definition~\ref{definition:flipflop_operator}) and 
cyclic operator (Definition~\ref{definition:cyclic_operator}) is derived from
Definition~\ref{definition:semigrouplike_operator}.

\paragraph{Flip-flop operator.}
Let us recall that the flip-flop monoid is $F = \{ r, s, e \}$ where $e$ is the
identity element, and $(r \cdot s) = s$, $(s \cdot s) = s$,
$(r \cdot r) = r$, $(s \cdot r) = r$.
We derive the operator defined by $F$, which has the form
\begin{align*}
  \langle X, T, \phi, \psi \rangle,
\end{align*}
and whose components we describe next.
First, we have that $X = F$.
For the transformations, note that multiplying an element in $F$ by $s$ and $r$
always yields $s$ and $r$, respectively, whereas multiplying an element by $e$
always yields the element itself.
These maps yield the transformations of the operator defined by $F$.
Specifically, we have that the set $T$ of transformations consists of the map 
$s(x) = s$, the map $r(x) = r$, and the map $e(x) = e$.
Thus, $T = F$, with the understanding that the elements of $F$ are seen as the
maps above.
The next step is a simplification step.
We observe that including the element $e$ in $X$ is redundant.
Note that element $e$ can be transformed to $s,r,e$ as $s(e) = s$,
$r(e) = r$, and $e(e) = e$, respectively.
At the same time, only $e$ itself can be transformed into $e$.
If we use the operator with initial element different from $e$---e.g.,
$\mathcal{F}(a,b \mid s)$---then the element $e$ will never be associated with
the operator.
If we use the operator with initial element $e$---e.g.,
$\mathcal{F}(a,b \mid e)$---then the associated element can be $e$ from the
beginning to a certain point, and then possibly become $r$ or $s$.
Both behaviours can be captured by a flip-flop operator that omits $e$ from $X$.
Thus, in our definition of flip-flop operator we omit $e$ from $X$.
Then, we can rename $s,r$ as elements of $X$ to $1,0$, respectively.
And we can rename $s,r,e$ as elements of $T$ to
$\mathit{set},\mathit{reset},\mathit{id}$, respectively.
Finally, the functions $\phi$ and $\psi$ are arbitrary, but importantly they are
surjective and injective, respectively, as required by the definition.

\paragraph{Cyclic operators.}
Let us recall that the cyclic group $C_n$ is (isomorphic to) the integers 
$\{ 0, \dots, n-1 \}$ with modular addition 
$i \cdot j = i + j \operatorname{mod} n$ as an operation.
We derive the operator defined by $C_n$, which has the form
\begin{align*}
  \langle X, T, \phi, \psi \rangle,
\end{align*}
and whose components we describe next.
First, we have that $X = \{ 0, \dots, n-1 \}$.
For the transformations, note that multiplying an element $i$ in $C_n$ by
another element $j$ yields $i+j \operatorname{mod} n$.
These multiplications yield the transformations of the operator defined by
$C_n$. Specifically, we have that the set $T$ of transformations consists of the
maps $i(j) = i \cdot j = i+j \operatorname{mod} n$ for every 
$i \in \{ 0, \dots, n-1\}$.
Thus, $T = \{ 0, \dots, n-1 \}$, with the understanding that each integer is
seen one the maps above.
Then, we can rename every map $i$ to $\operatorname{inc}_i$.
Finally, the functions $\phi$ and $\psi$ are arbitrary, but importantly they are
surjective and injective, respectively, as required by the definition.

\section{Additional Background}

\label{sec:additional-background}

We introduce additional background that is required in the proofs below.
We introduce  
(i)~background on Automata Theory, 
(ii)~the cascade composition of automata, 
(iii)~additional background on Semigroup and Group Theory, and 
(iv)~background on Algebraic Automata Theory.

\subsection{Automata}

We introduce \emph{finite automata}, c.f.\
\cite{ginzburg,arbib1969theories,hopcroft1979ullman,domosi2005algebraic}, which
here we term simply automata.
An automaton is a mathematical description of a stateful machine that processes
strings in a synchronous sequential manner, by performing an elementary
computation step on each input letter, and returning an output letter as a
result of the step. 
At its core lies the mechanism that updates the internal state at each step. 
This mechanism is captured by the notion of semiautomaton.

\paragraph{Syntax.}
A \emph{semiautomaton} is a tuple 
$D = \langle \Sigma, Q, \delta \rangle$ where:
$\Sigma$ is an alphabet called the \emph{input alphabet};
$Q$ is a finite non-empty set of elements called \emph{states}; 
$\delta: Q \times \Sigma \to Q$ is a function called \emph{transition function}.
A semiautomaton is trivial if $|Q| = 1$.

Automata are obtained from semiautomata by adding an initial state and output
function.

An \emph{automaton} is a tuple 
$A = \langle \Sigma, Q, \delta, q_\mathrm{init}, \Gamma, \theta \rangle$ where:
$D_A = \langle \Sigma, Q, \delta \rangle$ is a semiautomaton,
called the \emph{semiautomaton} of $A$;
$q_\mathrm{init} \in Q$ is called \emph{initial state};
$\Gamma$ is an alphabet called \emph{output alphabet};
$\theta: Q \to \Gamma$ is called \emph{output function}.
An \emph{acceptor} is a special kind of automaton where the output function is
a Boolean-valued function $\theta: Q \to \mathbb{B}$.
In the rest we assume that automata are acceptors.

It is worth noting that the kind of automaton we have defined is called a 
\emph{Moore machine}, as opposed to a Mealy machine.

\paragraph{Semantics.}
The \emph{execution} of automaton $A$ on an input string 
$s = \sigma_1 \dots \sigma_n$ proceeds as follows.
Initially, the automaton is in state $q_0 = q_\mathrm{init}$.
Upon reading letter $\sigma_i$, the state is updated as 
$q_i = \delta(q_{i-1},\sigma_i)$, and the automaton returns the letter
$\theta(q_i)$.
We denote by $A(s)$ the last letter returned by an automaton when executed on
$s$. 
An execution can be more formally stated by first extending the transition
function to strings.
The transition function is extended to the empty string $\varepsilon$ as
$\delta(q,\varepsilon) = q$ and to any non-empty string $\sigma_1 \dots
\sigma_n$ as $\delta(q,\sigma_1 \dots \sigma_n) =
\delta(\delta(q,\sigma_1), \sigma_2 \dots \sigma_n)$.
Then we have that $q_i = \delta(\sigma_1 \dots \sigma_i)$, and also that
$A(s) = \theta(\delta(q_\mathrm{init}))$.

The \emph{language recognised} by an automaton $A$ is the set of
strings $s$ such that $A(s) = 1$.
Two automata $A_1$ and $A_2$ are \emph{equivalent} if they recognise the same
language,
i.e., they are over the same input alphabet $\Sigma$, and the equality 
$A_1(s) = A_2(s)$ holds for every string $s \in \Sigma^*$.

\paragraph{Canonicity.}
A state $q \in Q$ is \emph{reachable} in an automaton $A$ from a state 
$q' \in Q$ if there exists an input string $\sigma_1 \dots \sigma_\ell$ such
that, for $q_0 = q'$ and $q_i = \delta(q_{i-1},\sigma_i)$, it holds that 
$q = q_\ell$.
Automaton $A$ is \emph{connected} if every state $q \in Q$ is reachable from the
initial state $q_0$.
For any state $q$ of automaton $A$, the automaton $A^q$ is the
automaton obtained by setting $q$ to be the initial state.
Two states $q$ and $q'$ of $A$ are \emph{equivalent} if the automata
$A^q$ and $A^{q'}$ are equivalent.
An automaton is in \emph{reduced form} if it has no distinct states which are
equivalent. 
An automaton is \emph{canonical} if it is connected and in reduced form.

\subsection{Semiautomata Cascades}

We introduce semiautomata cascades, c.f.\ Chapter~2 of \cite{domosi2005algebraic}.

A \emph{semiautomata cascade} is a semiautomaton
$A = \langle \Sigma, Q, \delta \rangle$
given by the composition of $d$ semiautomata
$A_1, \dots, A_d$ with connection functions 
$\phi_1, \dots, \phi_d$
as described next.
Every semiautomaton is of the form
$A_i = \langle \Sigma_i, Q_i, \delta_i \rangle$.
Every connection function is of the form $\phi_i: \Sigma \times Q_1 \times \dots
\times Q_{i-1} \to \Sigma_i$.
The state space of the cascade is
\begin{align*}
  Q & = Q_1 \times \dots \times Q_d,
\end{align*}
and its transition function is
\begin{align*}
  \delta(\langle q_1, \dots, q_d \rangle,\sigma) =
  \langle \delta_1(q_1, \sigma_1), \dots, \delta_d(q_d, \sigma_d) \rangle,
  \qquad
  \text{ where } \sigma_i = \phi_i(\sigma, q_1, \dots, q_{i-1}).
\end{align*}
We write the cascade as $(\phi_1, A_1) \ltimes \dots \ltimes 
(\phi_d, A_d)$, and we call each pair $(\phi_i, A_i)$ a \emph{component} of
the cascade.

\subsection{Additional Background on Semigroup and Group Theory}

We introduce additional background for Semigroup and Group Theory, cf.\
\cite{domosi2005algebraic,gallian2021contemporary}.

If $\psi$ is a homomorphism from a semigroup $S$ to a semigroup $T$ and 
$\psi$ is surjective, we say that $T$ is a \emph{homomorphic image} of $S$.
If $\psi$ is a homomorphism from a semigroup $S$ to a semigroup $T$ and 
$\psi$ is bijective, we call $\psi$ an \emph{isomorphism}. 
A semigroup $S$ \emph{divides} a semigroup $T$ if $T$ has a subsemigroup $T'$
such that $S$ is a homomorphic image of $T'$.

If $G$ is \emph{simple}, then every homomorphic image of $G$ is isomorphic to
$\{ e \}$ or $G$.

For $G$ and $H$ semigroups, we write 
$GH = \{ g\cdot h \mid g \in G, h \in H\}$. We also write 
$H^1 = H$ and $H^k = HH^{k-1}$.
A semigroup $S$ is \emph{generated} by a semigroup $H$ if 
$S = \bigcup_n H^n$; then $H$ is called a \emph{generator} of $S$.
For $G$ a group, and $N$ a normal subgroup of $G$, we call 
$G/N = \{ g \cdot N \mid g \in G \}$ a \emph{factor group}. Note that each 
$g \cdot N$ is a left coset.
A \emph{composition series} of a group $G$ is a sequence $G_0, \dots, G_n$ where 
$G_n = G$ and every $G_i$ is a maximal proper normal subgroup of $G_{i+1}$ for 
$i \in [0,n-1]$. Necessarily, $G_0$ is the trivial group. In standard notation, 
\begin{align*}
  \{ e \}  = G_0 \triangleleft G_1 \triangleleft \dots \triangleleft G_{n-1}
  \triangleleft G_n = G.
\end{align*}
For every $i \in [1,n]$,
the factor group $G_i/G_{i-1}$ is called a \emph{factor of the series}.
A \emph{solvable group} can be equivalently defined as a group with a composition
series all of whose factors are cyclic groups of prime order.

\subsection{Algebraic Automata Theory}

We introduce the necessary background on Algebraic Automata Theory, following
\cite{domosi2005algebraic}.

\paragraph{Connecting automata and semigroups.}
The set of transformations of a semiautomaton 
$\langle \Sigma, Q, \delta \rangle$
consists of every transformation $\delta_\sigma(q) = \delta(q,\sigma)$ for
$\sigma \in \Sigma$.
We consider semigroups consisting of transformations, where the
associative binary operation is function composition.
The \emph{characteristic semigroup} $\charsemigroup(A)$ of a semiautomaton
$A$ is the semigroup generated by its transformations. The characteristic
semigroup of an automaton is the one of its semiautomaton.
Given a semigroup $S$, the semiautomaton defined by $S$ is 
$A_S = \langle S, S, \delta \rangle$
where $\delta(s_1,s_2) = s_1 \cdot s_2$.
Note that $\charsemigroup(A_S) = S$.
A \emph{permutation semiautomaton} is a semiautomaton whose transformations
are permutations.
A \emph{group-free semiautomaton} is a semiautomaton such that only trivial
groups divide its characteristic semigroup. 

\paragraph{Homomorphic representations.}
\todo{Homomorphism may include inputs.}
A surjective function 
$\psi : Q \to Q'$ is a \emph{homomorphism}
of $A = \langle \Sigma, Q, \delta \rangle$ onto 
$A' = \langle \Sigma, Q', \delta' \rangle$ if, for every 
$\sigma \in \Sigma, q \in Q$, it holds that $\psi(\delta(q,\sigma)) =
\delta'(\psi(q),\sigma)$.
In this case, we say that $A'$ is a \emph{homomorphic image} of $A$.
If $\psi$ is bijective, then $A$ and $A'$ are said to be \emph{isomorphic}.
Given a semiautomaton $A = \langle \Sigma, Q, \delta \rangle$, a
semiautomaton $A' = \langle \Sigma', Q', \delta' \rangle$ is a 
\emph{subsemiautomaton} of $A$ if $\Sigma' \subseteq \Sigma$, $Q' \subseteq Q$,
and $\delta'$ is the restriction of $\delta$ to $\Sigma' \times Q'$ satisfies
$\delta'(q',\sigma') \in Q'$ for every $q' \in Q', \sigma' \in \Sigma'$.
If semiautomaton $A_1$ is the homomorphic image of a subsemiautomaton of a
semiautomaton $A_2$, then we say that $A_1$ is \emph{homomorphically
represented} by $A_2$.

\paragraph{Decomposition theorems.}
First we state the Jordan–H\"{o}lder decomposition theorem, which provides a way
to decompose semiautomata defined by groups. It is a rephrasing of Theorem~1.19
of \cite{domosi2005algebraic}, which is stated in terms of the wreath product.

\begin{theorem}[Jordan–H\"{o}lder decomposition theorem]
  \label{theorem:jordan_holder}
  If $G$ is a group, with composition series
  $G_0 \triangleleft G_1 \triangleleft \dots \triangleleft G_{n-1}
  \triangleleft G_n$,
  then the semiautomaton defined by $G$ is homomorphically represented by a
  semiautomata cascade $(\phi_1, A_1) \ltimes \dots \ltimes (\phi_n, A_n)$
  where $A_i$ is the semiautomaton defined by the simple group $G_i/G_{i-1}$.
\end{theorem}

Then we state the Krohn-Rhodes decomposition theorem, which applies to all
semiautomata---cf.\ Theorem~3.1 of \cite{domosi2005algebraic}.

\begin{theorem}[Krohn-Rhodes decomposition theorem]
  \label{theorem:krohn_rhodes}
  Let $A$ be a semiautomaton, and let $G_1, \dots, G_n$ be the simple groups
  that divide the characteristic semigroup $\charsemigroup(A)$.
  Then $A$ can be represented homomorphically by a semiautomata cascade 
  $(\phi_1, A_1) \ltimes \dots \ltimes (\phi_d, A_d)$
  where every $A_i$ is defined by some
  semigroup in $\{ F, G_1, \dots, G_n \}$ where $F$ is the flip-flop monoid.
  When $A$ is a permutation semiautomaton, then $F$ may be excluded.
Conversely, let $(\phi_1, A_1) \ltimes \dots \ltimes (\phi_d, A_d)$ be a
  semiautomata cascade that homomorphically represents a semiautomaton $A$. If a
  subsemigroup $S$ of the flip-flop monoid $F$ or a simple group $S$ divides
  $\charsemigroup(A)$, then $S$ divides $\charsemigroup(A_i)$ for some 
  $i \in [1,n]$. 
\end{theorem}

\medskip
\section{Preliminary Results}
\smallskip

\label{sec:appendix-preliminaries}

We prove two preliminary results (Proposition~\ref{prop:homomorphism} and
Proposition~\ref{prop:homomorphism_converse}) which allow us to take advantage of the
notion of homomorphic representation.
The results are novel, even though they share ideas with results that can be
found in the literature, e.g., \cite{ginzburg,arbib1969theories}.

\subsection{Homomorphic Representation implies Equivalence}

We state and prove Proposition~\ref{prop:homomorphism}.

\begin{restatable}{proposition}{prophomomorphism}
  \label{prop:homomorphism}
  If a semiautomaton $D_1$ homomorphically represents a semiautomaton 
  $D_2$, then every automaton with semiautomaton $D_2$ is equivalent to an
  automaton with semiautomaton $D_1$.
\end{restatable}
\begin{proof}
  Let us consider semiautomata $D_1$ and $D_2$.
  Assume that $D_1$ homomorphically represents $D_2$.
  It follows that there is a subsemiautomaton $D_1'$ of $D_1$ along with
  a homomorphism $\psi$ of $D_1'$ onto $D_2$. 
  Let 
  $D_1' = \langle \Sigma_1, Q_1, \delta_1 \rangle$
  and
  $D_2 = \langle \Sigma_2, Q_2, \delta_2 \rangle$,
  and let us consider an arbitrary acceptor $A_2$ having semiautomaton $D_2$,
  \begin{align*}
    A_2 = \langle \Sigma_2, Q_2, \delta_2, q^2_\mathrm{init}, \Gamma, \theta_2
    \rangle.
  \end{align*}
  Let us choose $q^1_\mathrm{init} \in Q_1$ such that 
  $\psi(q^1_\mathrm{init}) = q^2_\mathrm{init}$. 
  It exists because $\psi$ is surjective.
  We construct the automaton
  \begin{align*}
    A_1 = \langle \Sigma_2, Q_1, \delta_1, q^1_\mathrm{init}, \Gamma, \theta_1 
    \rangle,
  \end{align*}
  where the output function is
  $\theta_1(q) = \theta_2(\psi(q))$.
  Then we show that $A_1$ and $A_2$ are equivalent. 
  We consider the execution of $A_1$ and $A_2$ on an arbitrary input string
  $w = \sigma_1 \dots \sigma_n$. It suffices to show that 
  the last output of $A_1$ on $w$ equals the last output of $A_2$ on $w$.

  Let $q^1_0, \dots, q^1_n$ and $q^2_0, \dots, q^2_n$ be the 
  sequences of states for $A_1$ and $A_2$, respectively, in the considered
  execution on $w$.
  Namely, $q^1_0 = q^1_\mathrm{init}$ and 
  $q^1_i = \delta_1(q^1_{i-1}, \sigma_i)$ for $1 \leq i \leq n$.
  Similarly, $q^2_0 = q^2_\mathrm{init}$ and 
  $q^2_i = \delta_2(q^2_{i-1}, \sigma_i)$ for $1 \leq i \leq n$.

  As an auxiliary result, we show that
  $q^2_i = \psi_2(q^1_i)$ for every $0 \leq i \leq n$.
  We show it by induction on $i$.

  In the base case $i=0$ and $q^2_0 = \psi(q^1_0)$  amounts to
  $q^2_\mathrm{init} = \psi(q^1_\mathrm{init})$, which holds by definition.

  In the inductive case $i > 0$ and we assume that 
  $q^2_{i-1} = \psi(q^1_{i-1})$.
  We have to show $q^2_i = \psi(q^1_i)$. By the definition of $q^1_i$ and
  $q^2_i$ above, the former can be rewritten as 
  \begin{align*}
    \delta_2(q^2_{i-1},\sigma_i) = \psi(\delta_1(q^1_{i-1},\sigma_i)).
  \end{align*}
  Let us swap l.h.s.\ and r.h.s.,
  \begin{align*}
    \psi(\delta_1(q^1_{i-1},\sigma_i)) = \delta_2(q^2_{i-1},\sigma_i).
  \end{align*}
  Using the inductive hypothesis
  $q^2_{i-1} = \psi(q^1_{i-1})$, we obtain
  \begin{align*}
    \psi(\delta_1(q^1_{i-1},\sigma_i)) 
    = \delta_2(\psi(q^1_{i-1}),\sigma_i),
  \end{align*}
  which holds since $\psi$ is a homomorphism of $D_1$ onto $D_2$. This proves
  the auxiliary results.

  We have that 
  $A_1$ accepts $\sigma_1 \dots \sigma_n$ iff $\theta_1(q^1_n) = 1$, and
  similarly $A_2$ accepts $\sigma_1 \dots \sigma_n$ iff $\theta_2(q^2_n) = 1$.
  We have that 
  \begin{align*}
    \theta_1(q^1_n) =
    \theta_2(\psi(q^1_n)) = 
    \theta_2(q^2_n),
  \end{align*}
  where the first equality holds by definition of $\theta_1$, and the second
  equality holds by the auxiliary result.
  Therefore, $A_1$ and $A_2$ are equivalent.
  Now, $A_1$ has semiautomaton $D_1'$, rather than $D_1$ as required by the
  statement of the proposition.
  However, it is clear that $A_1$ can be extended to an equivalent automaton
  $A_1'$ that has semiautomaton $D_1$, since the states of $D_1$ that are not in
  $D_1'$ are not reachable from the initial state of $A_1$, and hence have no
  effect on the output.
  This proves the proposition.
\end{proof}

\subsection{Equivalence and Canonicity imply Homomorphic Representation}

In this section we state and prove Proposition~\ref{prop:homomorphism_converse}.
We first provide some preliminary definitions and propositions.

\begin{definition}
  Consider a language $L$ over $\Sigma$. 
  Two strings $x,y \in \Sigma^*$ are in relation $x \sim_L y$ if and only if the
  equivalence $(xz \in L) \Leftrightarrow  (yz \in L)$ holds for every 
  string $z \in \Sigma^*$.
  The resulting set of equivalence classes is written as $\Sigma^*/{\sim_L}$,
  and the equivalence class of a string $x \in \Sigma^*$ is written as $[x]_L$.
\end{definition}

\begin{definition}
  Consider an automaton $A$.
  Two input strings $x,y$ are in relation $x \sim_A y$ iff
  the state of $A$ upon reading $x$ and $y$ is the same.
  The resulting set of equivalence classes is written as $\Sigma^*/{\sim_A}$,
  and the equivalence class of a string $x \in \Sigma^*$ is written as $[x]_A$.
\end{definition}

\begin{proposition} \label{prop:statesim-implies-funcsim}
  For every automaton $A$ that recognises a language $L$,
  it holds that $x \sim_A y$ implies $x \sim_L y$.
\end{proposition}
\begin{proof}
  Assume $x \sim_A y$, i.e., the two strings lead to the same state $q$ of $A$.
  For every string $z$, the output of $A$ on both $xz$ and $yz$ is $A^q(z)$. 
  Hence, $A(xz) = A^q(z)$ and $A(yz) = A^q(z)$. 
  Since $A$ recognises $L$, it follows that $xz \in L \Leftrightarrow yz \in L$,
  and hence $x \sim_L y$.
\end{proof}

\begin{proposition} \label{prop:equivalence-classes-for-canonical-system}
  Consider a canonical automaton $A$ that recognises a language $L$ over
  an alphabet $\Sigma$.
  The states of $A$ are in a one-to-one correspondence with the equivalence
  classes $\Sigma^*/{\sim_L}$. 
  State $q$ corresponds to the equivalence class $[x]_L$ for $x$ any string that
  leads to $q$.
\end{proposition}
\begin{proof}
  Let $\psi$ be the mentioned correspondence.
First, $\psi$ maps every state to some equivalence class, since every state is
  reachable, because $A$ is canonical.
Second, $\psi$ maps every state to at most one equivalence class, by
  Proposition~\ref{prop:statesim-implies-funcsim}.
Third, $\psi$ is surjective since every equivalence class $[x]_L$ is assigned
  by $\psi$ to the the state that is reached by $x$.
Fourth, $\psi$ is injective since there are no distinct states
  $q,q'$ such that such that $[x]_L = [x']_L$ for $x$ leading to $q$ and $x'$
  leading to $q'$.
  The equality $[x]_L = [x']_L$ holds only if $x \sim_L x'$, which holds only if  
  $xz \in L \Leftrightarrow x'z \in L$ for every non-empty $z$, which holds only
  if $A^q = A^{q'}$ since $A$ recognises $L$. Since $A$ is canonical, and hence
  in reduced form, we have that $A^q = A^{q'}$ does not hold, and hence the
  $\psi$ is injective.
Therefore $\psi$ is a one-to-one correspondence as required.
\end{proof}

\begin{proposition} \label{prop:equivalence-classes-for-connected-system}
  Consider a connected automaton $A$ on input alphabet $\Sigma$.
  The states of $A$ are in a one-to-one correspondence with the equivalence
  classes $\Sigma^*/{\sim_A}$.
  State $q$ corresponds to the equivalence class $[x]_A$ for $x$ any string that
  leads to $q$.
\end{proposition}
\begin{proof}
  Let $\psi$ be the mentioned correspondence.
First, $\psi$ maps every state to some equivalence class, since every state is
  reachable, because $A$ is connected.
Second, $\psi$ maps every state to at most one equivalence class, by
  the definition of the equivalence classes $\Sigma^*/{\sim_A}$.
Third, $\psi$ is surjective since every equivalence class $[x]_A$ is assigned
  by $\psi$ to the state that is reached by $x$.
Fourth, $\psi$ is injective since there are no distinct states $q,q'$ such
  that $[x]_A = [x']_A$ for $x$ leading to $q$ and $x'$ leading to $q'$.
  The equality $[x]_A = [x']_A$ holds only if $x \sim_A x'$, which holds only if  
  $x$ and $x'$ lead to the same state.
Therefore $\psi$ is a one-to-one correspondence as required.
\end{proof}

\begin{restatable}{proposition}{prophomomorphismconverse}
  \label{prop:homomorphism_converse}
  If an automaton $A_1$ is equivalent to a canonical automaton $A_2$, then the
  semiautomaton of $A_1$ homomorphically represents the semiautomaton of $A_2$.
\end{restatable}
\begin{proof}
  Consider an automaton $A_1$ and a canonical automaton $A_2$.
  \begin{align*}
    A_1 & = \langle \Sigma, Q_1, \delta_1, q_1^\mathrm{init}, \Gamma, \theta_1
    \rangle
    \\
    A_2 & = \langle \Sigma, Q_2, \delta_2, q_2^\mathrm{init}, \Gamma, \theta_2
    \rangle
  \end{align*}
  Assume that the two automata are equivalent, i.e., they recognise the same
  language $L$.
  By Proposition~\ref{prop:equivalence-classes-for-canonical-system},
  every state of $A_2$ can be seen as an equivalence class $[x]_L$.
  Let $A_1'$ be the reachable subautomaton of $A_1$, and let $D_1'$ be its
  semiautomaton.
  \begin{align*}
    D_1' & = \langle \Sigma, Q_1', \delta_1' \rangle
  \end{align*}
  By Proposition~\ref{prop:equivalence-classes-for-connected-system}, every
  state of $A_1'$ can be seen as an equivalance class $[x]_{A_1}$.
  Let us define the function $\psi$ that maps $[x]_{A_1}$ to $[x]_L$.
  We have that $\psi$ is a well-defined function, i.e., it does not assign
  multiple values to the same input, by
  Proposition~\ref{prop:statesim-implies-funcsim} since $A_1$ recognises $L$.

  We argue that $\psi$ is a surjective function as required by the definition of
  homomorphism.
  The function is surjective since every state $q$ in $A_2$ is reachable, hence
  there is $x$ that reaches it, and hence $\psi$ maps $[x]_{A_1}$ to 
  $[x]_L = q$.

  Having argued the properties above, in order to show that $\psi$ is
  a homomorphism from $D_1'$ to $D_2$, it suffices to show that, for every $q
  \in Q_1'$ and $\sigma \in \Sigma$, the following equality holds.
  \begin{align*}
    \psi\big(\delta_1'(q,\sigma)\big) = \delta_2\big(\psi(q),\sigma\big)
  \end{align*}
  Let us consider arbitrary
  $q \in Q_1'$ and $\sigma \in \Sigma$.
  Let $x \in \Sigma^*$ be a string that reaches $q$ in $A_1'$.
  Note that $q$ can be seen as the equivalence class $[x]_{A_1}$.
  We have the following equivalences:
  \begin{align*}
    & \psi(\delta_1'(q,\sigma)) = \delta_2(\psi(q),\sigma)
    \\
    & \Leftrightarrow \psi(\delta_1'([x]_{A_1},\sigma)) =
    \delta_2(\psi([x]_{A_1'}),\sigma)
    \\
    & \Leftrightarrow \psi([x\sigma]_{A_1}) = \delta_2(\psi([x]_{A_1}),\sigma)
    \\
    & \Leftrightarrow [x\sigma]_L = \delta_2(\psi([x]_{A_1}),\sigma)
    \\
    & \Leftrightarrow [x\sigma]_L = \delta_2([x]_L,\sigma)
    \\
    & \Leftrightarrow [x\sigma]_L = [x\sigma]_L.
  \end{align*}
  The last equality holds trivially, and hence
  $\psi$ is a homomorphism from $D_1'$ to $D_2$. Since $D_1'$ is
  a subsemiautomaton of $D_1$, we conclude that $D_1$ homomorphically represents
  $D_2$.
  This proves the proposition.
\end{proof}

\iflogicprograms
  \subsection{Prime Decomposition of Transformation Logic Programs}
\else
  \subsection{Prime Decomposition Theorem for the Transformation Logics}
\fi

We introduce a decomposition theorem for the Transformation Logics, along the
lines of the Krohn-Rhodes decomposition theorem.
The theorem is in two parts. The first part
(Theorem~\ref{theorem:prime_decomposition_of_transformation_logics}) shows how
to decompose a transformation \wrule{} into a \wprogram{} over prime operators.
The second part
(Theorem~\ref{theorem:prime_decomposition_of_transformation_logics_converse})
shows that prime operators are indeed prime, i.e., they cannot be decomposed
into other prime operators.

We first introduce some preliminary definitions.

\begin{definition}
  Given an operator $\mathcal{T} = \langle X, T, \phi, \psi \rangle$,
  its \emph{characteristic semigroup} $\mathbf{S}(\mathcal{T})$ 
  is the semigroup generated by $T$.
\end{definition}

\begin{definition}
  An operator is a \emph{permutation operator} if all its transformations are
  bijections.
\end{definition}

\begin{definition}
  A \wprogram{} $P_1$ is captured by a \wprogram{} $P_2$ if the \wprogram{}s have
  the same set of input variables and, for every variable $q$ defined in $P_1$,
  every input $I$ to $P_1$, and every index $t$ over the positions of $I$, 
  it holds that $(P_1, I, t) \models q$ iff $(P_2, I, t) \models q$.
  A \wrule{} $r$ is captured by a \wprogram{} $P$ if the singleton \wprogram{} 
  $\{ r \}$ is captured by $P$.
\end{definition}

\begin{restatable}{theorem}{theoremkrohnrhodes}
  \label{theorem:prime_decomposition_of_transformation_logics}
  Let $r$ be a transformation \wrule{} with operator $\mathcal{T}$
  and let $G_1, \dots, G_n$ be the simple groups
  that divide $\mathbf{S}(\mathcal{T})$.
  Then $r$ is captured by a \wprogram{} with prime operators defined by
  $\{ F, G_1, \dots, G_n \}$ where $F$ is the flip-flop monoid.
  Moreover, if $\mathcal{T}$ is a permutation operator,
  then $F$ may be excluded.
\end{restatable}
 \begin{proof}[Proof idea]
   We start by mapping the given \wrule{} $r$ to an automaton $A$.
   Then, we obtain a prime decomposition $A'$ for $A$, using the Krohn-Rhodes
   decomposition theorem for automata.
   Finally, we map $A'$ to a \wprogram{}, making sure every component of the cascade
   maps to a transformation operator with the same characteristic semigroup.
 \end{proof}
\begin{proof}
  Let $\mathbf{P} = \{ G_1, \dots, G_n \}$ if $\mathcal{T}$ is a permutation
  operator, and $\mathbf{P} = \{ F, G_1, \dots, G_n \}$ otherwise.
  Let $\mathcal{T} = \langle X, T, \phi, \psi \rangle$, and let $m$ and $n$ be
  its input and output arity, respectively.
  Let \wrule{} $r$ be of the form
  \begin{align*}
    p_1, \dots, p_n \ruleif \mathcal{T}(a_1, \dots, a_m \mid x_0).
  \end{align*}
  Let $V = \{a_1, \dots, a_m \}$ and let $U = \{ p_1, \dots, p_n \}$.
  Let $P$ be the singleton \wprogram{} $\{ r \}$.
  We construct an automaton $A$ that is equivalent to $P$ and it has the same
  characteristic semigroup as $\mathcal{T}$.
Automaton $A$ is defined as $\langle \Sigma, Q, \delta, q_\mathrm{init}, \Gamma, \theta \rangle$ where 
  the input alphabet is $\Sigma = 2^V$,
  the set of states is $Q = X$, 
  the transition function is defined as $\delta(q,\sigma) = \tau(q)$ with
  $\tau = \phi(b_1, \dots, b_m)$ where $\langle b_1, \dots, b_m \rangle$ is the
  indicator vector of the set of variables $\sigma$, 
  the initial state is $q_\mathrm{init} = x_0$,
  the output alphabet is $\Gamma = 2^U$, and
  the output function $\theta(x)$ returns the set of variables indicated by
  $\psi(x)$. \todo{define what it means that a vector indicates a set.}
We have that $A$ is equivalent to $P$ in the sense that, for every input 
  $I = I_1, \dots, I_\ell$ to $P$ and every $t \in [1,\ell]$, it holds that 
  $(P,I,t) \models p_i$ iff $p_i \in A(I_{1:t})$ with 
  $I_{1:t} = I_1 \dots I_t$.
  \todo{A proof of equivalence here would be required.}

  By the Krohn-Rhodes decomposition theorem
  (Theorem~\ref{theorem:krohn_rhodes}), there is a semiautomata cascade
  $C = (\phi_1, A_1) \ltimes \dots \ltimes (\phi_d, A_d)$ that
  homomorphically represents the semiautomaton of $A$ such that 
  every $A_i$ is defined by a semigroup in $\mathbf{P}$.

  By Proposition~\ref{prop:homomorphism}, there exists an automaton 
  $A' = \langle \Sigma, Q_C, \delta_C, q_\mathrm{init}^C, \Gamma, \theta'
  \rangle$ with semiautomaton $C$
  such that $A'$ is equivalent to $A$, and hence to $P$.
  In particular, its states are 
  $Q_C = Q_1 \times \dots \times Q_d$
  and its initial state is 
  $q_\mathrm{init}^C = \langle q_\mathrm{init}^1, \dots, q_\mathrm{init}^d
  \rangle$.
  We construct a \wprogram{} $P'$ that is equivalent to $A'$ (hence to $P$) and it
  has operators defined by $\mathbf{P}$.

  First, by induction on the number of components $d$ of the cascade $C$, we
  construct a \wprogram{} $P_C$ that tracks the state of the casacade $C$, in the
  following sense. 
  Let $Q_i = \{ q_{i,1}, \dots, q_{i,m_i} \}$ for every $i \in [1,d]$.
  We use every state $q_{i,1}$ also as a propositional variable to indicate the
  state.
  Then, 
  $P_C$ tracks the state of the casacade $C$ in the sense that
  $(P_C, I, t) \models q_{i,j}$ iff $q_{i,j}$ is the $i$-th component of the
  state of the cascade $C$ when executed on input $I_{1:t}$.

  In the base case we have $C = (\phi_1, A_1)$.
  Let us recall that the input alphabet of the cascade is $\Sigma = 2^V$.
  Also that $A_1$ is semigrouplike, and hence it is of the form
  $A_1 = \langle Q_1, Q_1, \delta_1 \rangle$.
  Thus $\phi_1$ is of the form $\phi_1 : \Sigma \to Q_1$.
Let $\mathcal{A}_1 = \langle Q_1, T_1, \xi_1, \psi_1 \rangle$
  be the semigrouplike operator defined by the characteristic semigroup of
  $A_1$.
  We have that $T_1$ coincides with the set of transformations of $A_1$.
  Note that the functions $\xi_1$ and $\psi_1$ are given, as discussed in the
  main body. In particular we, have that 
  $\xi_1 : \mathbb{B}^{k_1} \to Q_1$ is surjective, 
  and $\psi_1 : Q_1 \to \mathbb{B}^{\ell_1}$ is injective.
Now, every transformation of $T_1$ corresponds to an element $x \in Q_1$.
  Since $\xi_1$ is surjective, for every $x \in Q_1$, there exists a
  binary vector $\vec{b}$ such that $\xi_1(\vec{b}) = x$. Fixing an arbitrary
  choice of the binary vector $\vec{b}$ to return for $x$, let us define the
  `inverse' function $\xi_1^{-1}(x) =\vec{b}$.
  Thus, we can define the composition $f = \phi_1 \circ \xi_1^{-1}$ which is a
  function of the form $f : \Sigma \to \mathbb{B}^{k_1}$.
Thus, on input $\sigma \in \Sigma$, semiautomaton $A_1$ receives 
  $\phi_1(\sigma) = x$ and it applies the transformation $x$,
  and also the operator $\mathcal{A}_1$ applies
  the same transformation $x$ since
  $(\phi_1 \circ \xi^{-1} \circ \xi)(\sigma) = \phi_1(\sigma) = x$.
Furthermore, $f$ is a function from sets of variables
  $\Sigma = 2^V$ to Boolean tuples, and hence each
  of its projetions $f_i$ can be seen as Boolean function, which can be
  represented by a propositional formula $\varphi_{1,i}$ over variables $V$.
  Let $a_{1,1}, \dots, a_{1,k_1}$ and $b_{1,1}, \dots, b_{1,k_1}$ be fresh
  variables.
  We introduce the static \wrule{}s $r_{1,i}$ defined as
  \begin{align*}
    a_{1,i} & \ruleif \varphi_{1,i}
  \end{align*}
  and the transformation \wrule{} $r_{A_1}$ defined as
  \begin{align*}
    b_{1,1}, \dots, b_{1,k_1} \ruleif 
    \mathcal{A}_1(a_{1,1}, \dots, a_{1,k_1} \mid q_\mathrm{init}^1).
  \end{align*}
Now, based on $\psi_1$, we can introduce \wrule{}s for the variables 
  $Q_1 = \{ q_{1,1}, \dots, q_{1,m_1} \}$.
  Since $\psi_1$ is injective,
  there exists a map $g_1 : \mathbb{B}^{\ell_1} \to \mathbb{B}^{m_1}$ such that
  $g_1$, for every $q_{i,j} \in Q_1$, we have that $g_1(\psi_1(q_{i,j}))$ is the
  one-hot encoding of $q_{i,j}$ within $Q_1$.
  We introduce the static \wrule{}s $r_{q_{1,i}}$ defined as
  \begin{align*}
    q_{1,i} & \ruleif \beta_{1,i}
  \end{align*}
  where $\beta_{1,i}$ is a propositional formula that represents the $i$-th
  projection of $g_1$.
Then we have that the \wprogram{} 
  $P_C = \{ r_{A_1}, r_{a_{1,1}}, \dots,  r_{a_{1,k_1}}, r_{q_{1,1}}, \dots,
  r_{q_{1,m_1}} \}$ tracks the state of the cascade $C$.
  \todo{A proof of equivalence here would be required.}

  In the inductive case we have
  $C = (\phi_1, A_1) \ltimes \dots \ltimes (\phi_{d-1}, A_{d-1}) \ltimes 
  (\phi_d, A_d)$ with $d \geq 2$, and 
  we assume that we have a \wprogram{} $P_{C'}$ for the cascade 
  $C' = (\phi_1, A_1) \ltimes \dots \ltimes (\phi_{d-1}, A_{d-1})$.
  Recall that the state of the cascade $C'$ is represented in $P_{C'}$ by
  variables $q_{i,j}$.
  We have that every component in a cascade reads the state of the previous
  components before it is updated. In order to capture this aspect,
  we introduce:
  \par\smallskip\noindent
  (i) fresh variables $Q_i' = \{ q'_{i,1}, \dots, q'_{i,m_i} \}$ and
  the delay \wrule{}s 
  \begin{align*} 
    q_{i,j}' \ruleif \operatorname{\mathcal{D}} q_{i,j}
  \end{align*}
  (ii) fresh variables $Q_i'' = \{ q''_{i,1}, \dots, q''_{i,m_i} \}$ and the
  following static \wrule{}s for $i,j$ satisfying $q_{i,j} \neq q_\mathrm{init}^i$,
  \begin{align} \label{eq:state-definitions-1}
    q_{i,j}'' \ruleif q_{i,j}'
  \end{align}
  (iii) the following static \wrule{}s for $i,j$ satisfying 
  $q_{i,j} = q_\mathrm{init}^i$,
  \begin{align} \label{eq:state-definitions-2}
    q_{i,j}'' \ruleif q_{i,j}' \lor (\neg q_{1,i}', \land \dots \land \neg
    q_{i,m_i}').
  \end{align}
  Note that, at the first time point, every $q_{i,j}'$ is false by the semantics
  of the delay operator, and hence \wrule{}s~\eqref{eq:state-definitions-1}
  and~\eqref{eq:state-definitions-2} ensure that $q_{i,j}''$ holds iff 
  $q_{i,j} = q_\mathrm{init}^i$. At later time points, some $q_{i,j}'$ holds,
  and hence \wrule{}~\eqref{eq:state-definitions-2} is triggerd only if its
  first disjunct $q_{i,j}'$ holds.
The rest of the inductive case proceeds similarly to the base case, with the
  only difference that the relevant variables are now 
  $V \cup Q_1'' \cup \dots \cup Q''_{d-1}$, where $V$ are the input variables,
  and $Q_1 \cup \dots \cup Q_{d_1}$ track the state of $C'$ before it is
  updated.
  The additional variables have no impact on the argument, as they can now be
  seen as input variables.
  Thus, following the same arguments as above, we introduce the static \wrule{}s 
  \begin{align*}
    a_{d,i} & \ruleif \varphi_{d,i}
  \end{align*}
  the transformation \wrule{} 
  \begin{align*}
    b_{d,1}, \dots, b_{d,k_d} \ruleif 
    \mathcal{A}_d(a_{d,1}, \dots, a_{d,k_d} \mid q_\mathrm{init}^d),
  \end{align*}
  and the static \wrule{}s  
  \begin{align*}
    q_{d,i} & \ruleif \beta_{d,i}
  \end{align*}
  Then we have that the \wprogram{} $P_C$ obtained by extending $P_{C'}$ with the
  \wrule{}s above tracks the state of the cascade $C$. 
  \todo{A proof of equivalence would be required.}

  Finally, to obtain the \wprogram{} $P'$ that captures the initial \wprogram{} $P$, it
  is enough to capture the output function $\theta'$ of the automaton $A'$.
  We have that the function is of the form $\theta: Q_1 \times \dots \times Q_d
  \to \Gamma$ where $\Gamma = 2^U$ and $U = \{ p_1, \dots, p_n \}$.
  We construct one static \wrule{} $r_{p_i}$ for every $p_i$ as
  \begin{align*}
    p_i \ruleif \alpha_i
  \end{align*}
  where $\alpha_i$ is a propositional formula on variables 
  $Q_1 \cup \dots \cup Q_d$ that is true iff 
  $p_i \in \theta(q_1, \dots, q_d)$, and it exists because every
  Boolean function admits a propositional formula.
  \todo{We could clarify how the variables capture the input to the functions.}
  Then the \wprogram{} $P' = P_C \cup \{ r_{p_1}, \dots, r_{p_n} \}$ captures $P$. 
  \todo{A proof of equivalence would be required.}
\end{proof}

\begin{restatable}{theorem}{theoremkrohnrhodesconverse}
  \label{theorem:prime_decomposition_of_transformation_logics_converse}
  If a transformation \wrule{} $r$
  with prime operator $\mathcal{P}$ is captured by
  a \wprogram{} with operators $\{ \mathcal{T}_1, \dots,
  \mathcal{T}_n \}$, then $\mathbf{S}(\mathcal{P})$ divides one of the
  semigroups in $\{ F, \mathbf{S}(\mathcal{T}_1), \dots,
  \mathbf{S}(\mathcal{T}_n) \}$ where $F$ is the flip-flop monoid. 
\end{restatable}
\begin{proof}[Proof idea]
  We start by mapping the given \wrule{} $r$ to an automaton $A$.
  Second, we map the given \wprogram{} $P$ to an equivalent cascade $C$, ensuring
  that every operator maps to a component having the same characteristic
  semigroup.
  In particular, we have that $C$ captures $A$, and hence 
  the semigroup of $A$ divides the semigroup of a single component of $C$,
  by the converse of the Krohn-Rhodes theorem.
  We conclude that the characteristic semigroup of $\mathcal{P}$ divides the
  characteristic semigroup of some operator $\mathcal{T}_i$.
\end{proof}
\begin{proof}
  Let \wrule{} $r$ be of the form
  \begin{align*}
    p_1, \dots, p_n \ruleif \mathcal{P}(a_1, \dots, a_m \mid x_0).
  \end{align*}
  Let $\mathcal{P} = \langle X, T, \phi, \psi \rangle$, and let $m$ and $n$ be
  its input and output arity, respectively.
  Let $V = \{a_1, \dots, a_m \}$ and let $U = \{ p_1, \dots, p_n \}$.
  Let $P$ be the singleton \wprogram{} $\{ r \}$.
  First we construct an automaton $A$ that is equivalent to $P$ and it has the
  same characteristic semigroup as $\mathcal{P}$.
Automaton $A$ is defined as 
  $\langle \Sigma, Q, \delta, q_\mathrm{init}, \Gamma, \theta \rangle$ where 
  the input alphabet is $\Sigma = 2^V$,
  the set of states is $Q = X$, 
  the transition function is defined as $\delta(q,\sigma) = \tau(q)$ with
  $\tau = \phi(b_1, \dots, b_m)$ where $\langle b_1, \dots, b_m \rangle$ is the
  indicator vector of the set of variables $\sigma$, 
  the initial state is $q_\mathrm{init} = x_0$,
  the output alphabet is $\Gamma = 2^U$, and
  the output function $\theta(x)$ returns the set of variables indicated by
  $\psi(x)$. \todo{define what it means that a vector indicates a set.}
We have that $A$ is equivalent to $P$ in the sense that, for every input 
  $I = I_1, \dots, I_\ell$ to $P$ and every $t \in [1,\ell]$, it holds that 
  $(P,I,t) \models p_i$ iff $p_i \in A(I_{1:t})$ with 
  $I_{1:t} = I_1 \dots I_t$.
  \todo{A proof of equivalence here would be required.}

  Let $P'$ be the \wprogram{} that captures $P$, with operators in 
  $\{ \mathcal{T}_1, \dots, \mathcal{T}_n \}$.
  We build a semiautomata casacade 
  $C = (\phi_1, A_1) \ltimes \dots \ltimes (\phi_d, A_d)$ that captures the
  variables defined by transformation and delay \wrule{}s in $P'$.

  \todo{define what it means to capture the variables.}

  \todo{introduce the idea of a \emph{reactive semiautomata cascade}}
  We proceed by induction on the number of transformation and delay \wrule{}s in
  $P'$.
  In the base case there is no delay or transformation \wrule{}, and hence the claim
  holds trivially.
  In the inductive case, we 
  consider the \wprogram{} $P''$ obtained by removing a transformation or delay
  \wrule{} $r'$ that has no other transformation or delay \wrule{} depending on it.
  \todo{define dependency}
  We assume by induction that we have a cascade 
  $C' = (\phi_1, A_1) \ltimes \dots \ltimes (\phi_{d-1}, A_{d-1})$ that captures the
  variables defined by transformation or delay \wrule{}s in $P''$.

  We extend it to a cascade 
  $C = (\phi_1, A_1) \ltimes \dots \ltimes (\phi_{d-1}, A_{d-1}) \ltimes
  (\phi_d, A_d)$ 
  that captures $P'$.
  We consider two cases separately.

  In the first case, we have that $r'$ is a delay \wrule{}
  \begin{align*}
    p \ruleif \operatorname{\mathcal{D}} a.
  \end{align*}
  Let $f$ be the function that returns the truth value of $a$ based on the
  variables defined by delay or transformation \wrule{}s in $P''$.
  Then we have that $\phi_d(\sigma, q_1, \dots, q_{d-1})$ returns $\mathit{set}$
  if $f(x) = 1$ and $\mathit{reset}$ otherwise.
  Then $A_d$ is the flip-flop semiautomaton.
  Thus, we have that $C$ and in particular its last component $(\phi_d, A_d)$
  captures the variables $p$ defined by the delay \wrule{} since $q_d$ will be the
  previous truth value of variable $a$.

  In the second case, we have that $r'$ is a transformation \wrule{}
  \begin{align*}
    p_1, \dots, p_n \ruleif \operatorname{\mathcal{T}}(a_1, \dots, a_n \mid
    x_0).
  \end{align*}
  Automaton $A_d$ is chosen to be the semiautomaton defined by the
  characteristic semigroup of $\mathcal{T}$.
  Function $\phi_d(\sigma, q_1, \dots, q_{d-1})$ is chosen as follows.
  First, let us note that the value of variables $a_1, \dots, a_n$ is determined
  by some functions $f_1, \dots, f_n$ of the input variables $\sigma$ and of the
  variables $b_1, \dots, b_{d-1}$ defined by delay or transformation \wrule{}s in
  $P''$. 
  In particular, $f_1, \dots, f_n$ are defined by static \wrule{}s in $P''$, and
  they can be built by induction on such \wrule{}s. 
Second, let us note that the before-update states $q_1, \dots, q_{d-1}$ capture
  the value of variables $b_1, \dots, b_{d-1}$ at the previous time point, and
  hence the after-update states $q_1', \dots, q_{d-1}'$ capture the value of
  variables $b_1, \dots, b_{d-1}$ at the current time point.
Now, let $\phi: \mathbb{B}^n \to T$ be the decoding function of the
  transformation operator in the considered \wrule{}, and let us write
  $\mathbf{x}$ for the list of variables $\sigma, q_1', \dots, q_{d-1}'$.
  Then, the value of variables $a_1, \dots, a_n$ is functionally determined as 
  $f_1(\mathbf{x}), \dots, f_n(\mathbf{x})$
  starting from the input $\sigma$ and the states $q_i'$. In turn, every state
  $q_i'$ is a function of $\sigma, q_1, \dots, q_{d-1}$. In fact, if we let
  $\delta_i$ be the transition function of $A_i$ for $i \in [1,d-1]$,
  we have that every $q_i'$ is defined inductively as 
  $q'_1 = \delta_1(q_1,\sigma)$ in the base case and as 
  $q'_i = \delta_i(q_i,\langle \sigma, q'_1,\dots,q'_{d-1} \rangle)$ 
  for the inductive case $i \in [2,d-1]$.
  Therefore, we can choose $\phi_d(\sigma, q_1, \dots, q_{d-1}) = 
  \phi(f_1(\mathbf{x}), \dots, f_n(\mathbf{x}))$. 

  We have that the automaton obtained from $A$ with identity output and initial
  state $x_0$ is canonical. \todo{We need a proposition saying that automata
    defined by primes are canonical. The flip-flop is canonical by inspection.
    The simple groups are canonical (i) they are reachable because of
    permutations, (ii) any two elements are distinguishable in one step because
    of the invertibility requirement, they cannot yield the same element.}
    Then, by Proposition~\ref{prop:homomorphism_converse}, we have that $C$
    homomorphically represents $A$.
  Since $\mathcal{P}$ is prime, we have that the characteristic semigroup $S$ of
  $\mathcal{P}$ is the flip-flop monoid or a simple group.
  We have that $S$ is also the characteristic semigroup of $A$.
  By the converse of the Krohn-Rhodes decomposition theorem
  (Theorem~\ref{theorem:krohn_rhodes}), it follows that $S$ divides the
  semigroup of some $A_i$, which is the characteristic semigroup of some
  $\mathcal{T}_i$.
\end{proof}

\section{Proofs}

\label{sec:proofs-expressivity}

We provide the proofs of all our results in the same order as they are
stated in the main body.

\subsection{Proof of Theorem~\ref{theorem:expressivity_transformation_logics}}

\theoremexpressivitytransformationlogics*
\begin{proof}
  We prove the theorem in two parts.

  First, to show that every language recognised by a \wprogram{} is regular, it
  suffices to show that every \wprogram{} admits an equivalent automaton.
  This has already been shown in the proof of
  Theorem~\ref{theorem:prime_decomposition_of_transformation_logics_converse}.

  Then, to show that every regular language is captured by a transformation
  logic, it
  suffices to show that every automaton acceptor can be turned into a
  \wprogram{}.
  Let us consider an automaton $A = \langle \Sigma, Q, \delta, q_\mathrm{init},
  \Gamma, \theta \rangle$.
  Let $T$ consist of each transformation $\delta_\sigma(q) = \delta(q,\sigma)$
  for $\sigma \in \Sigma$.
  Let $\Sigma = \{ \sigma_1, \dots, \sigma_m \}$, and 
  let $Q = \{ q_1, \dots, q_n \}$.
  Let us consider $\sigma_1, \dots, \sigma_m$ and 
  $q_1, \dots, q_n$ also as propositional variables.
  Let $\phi(\vec{b})$ be the function that maps every one-hot vector 
  $\vec{b} \in \mathbb{B}^m$ to the corresponding transformation 
  $\delta_{\sigma_i}$.
  Let $\psi(\vec{b})$ be the function that maps every state $q \in Q$ to its
  one-hot encoding.
  Let us define the transformation operator 
  $\mathcal{T} = \langle Q, T, \phi, \psi \rangle$.
  We introduce the transformation \wrule{} $r_\mathrm{t}$ defined as
  \begin{align*}
    q_1, \dots, q_n \ruleif \mathcal{T}(\sigma_1, \dots, \sigma_m \mid
    q_\mathrm{init})
  \end{align*}
  and the static \wrule{} $r_\mathrm{out}$ defined as
  \begin{align*}
    \mathit{out} \ruleif \alpha
  \end{align*}
  where $\alpha$ is a propositional formula over variables $Q$ which
  is true iff $\theta(q) = 1$.
  Then we have that the \wprogram{} $P = \{ r_\mathrm{t}, r_\mathrm{out}\}$ is
  equivalent to the given automaton acceptor $A$ in the sense that,
  for every string $s = \sigma_{i_1} \dots \sigma_{i_\ell}$,
  we have that $A(s) = 1$ iff $(P, I, \ell) \models \mathit{out}$ where
  $I = \{ \sigma_{i_1} \}, \dots, \{ \sigma_{i_\ell} \}$.
\end{proof}

\subsection{Proof of Theorem~\ref{theorem:expressivity_transformation_logics_prime}}

\theoremexpressivitytransformationlogicsprime*
\begin{proof}
  We prove the theorem in two parts.

  First, the expressivity is within the regular languages by
  Theorem~\ref{theorem:expressivity_transformation_logics}.

  Then, again by Theorem~\ref{theorem:expressivity_transformation_logics}, we
  know
  that every regular language is recognised by some \wprogram{} $P$.
  By Theorem~\ref{theorem:prime_decomposition_of_transformation_logics}, we can
  replace every transformation
  \wrule{} in $P$ with a \wprogram{} over prime operators. By replacing all
  \wrule{}s
  with non-prime operators, we obtain the desired \wprogram{}.
\end{proof}

\subsection{Proof of Theorem~\ref{theorem:expressivity_kr_logics_starfree}}

\theoremexpressivitykrlogicsstarfree*
\begin{proof}
  We prove the theorem in two parts.

  First, we show that the expressivity of $\mathcal{L}(\mathcal{F})$ is
  contained in the star-free regular languages. 
  Consider a query $(P,q)$ of $\mathcal{L}(\mathcal{F})$, and let $L$ be the
  language it recognises.
  Following the construction given in the proof of
  Theorem~\ref{theorem:prime_decomposition_of_transformation_logics_converse},
  we can build a cascade $A$ of flip-flops that is equivalent to $(P,q)$, and
  hence it recognises $L$.
  By Lemma~C of Section~7.12 of \cite{ginzburg}, every language recognised by a
  cascade of flip-flops is star-free, and hence $L$ is
  star-free.
  Since $(P,q)$ is an arbitrary query of $\mathcal{L}(\mathcal{F})$, we conclude
  that the expressivity of $\mathcal{L}(\mathcal{F})$ is contained in the
  star-free regular languages. 

  Then, we show that the expressivity of $\mathcal{L}(\mathcal{F})$ contains
  the star-free regular languages.
  Consider a star-free regular language $L$.
  By the correspondence between star-free and noncounting regular
  languages \cite{schutzenberger1965finite}, there exists a noncounting
  automaton.
  By Lemma~A of Section~7.12 of \cite{ginzburg}, there exists a group-free
  automaton $A$, i.e., no non-trivial group divides $\mathbf{S}(A)$, the
  characteristic semigroup of $A$.
  Then there exists a \wprogram{} $P$ with operator $\mathcal{A}$ defined by 
  $\mathbf{S}(A)$ that recognises $L$---the construction of $P$ proceeds as in
  the proof of Theorem~\ref{theorem:expressivity_transformation_logics}.
  Then by Theorem~\ref{theorem:prime_decomposition_of_transformation_logics},
  there exists a \wprogram{} $P'$ whose
  only operator is $\mathcal{F}$.
  Since $L$ is an arbitrary star-free regular language, we conclude that 
  the expressivity of $\mathcal{L}(\mathcal{F})$ contains the star-free regular
  languages.
\end{proof}

\subsection{Proof of Theorem~\ref{theorem:groups_add_expressivity}}

\theoremgroupsaddexpressivity*
\begin{proof}
  First, we have that the expressivity of $\mathcal{L}(\mathcal{F}, \mathbf{G})$
  contains the expressivity of $\mathcal{L}(\mathcal{F})$.
  Then, to prove the theorem, it suffices to show
  a \wprogram{} of $\mathcal{L}(\mathcal{F}, \mathbf{G})$ that is not captured by 
  any \wprogram{} of $\mathcal{L}(\mathcal{F})$.

  Let us consider a group operator $\mathcal{G} \in \mathbf{G}$, and a
  transformation \wrule{} $r$ of the form 
  \begin{align*}
    p_1, \dots, p_n \ruleif \mathcal{G}(a_1, \dots, a_m \mid x_0).
  \end{align*}
  Then, let us consider an arbitrary \wprogram{} $P$ of $\mathcal{L}(\mathcal{F})$.
  Let us assume by contradiction that $P$ captures $\{ r \}$.
  By
  Theorem~\ref{theorem:prime_decomposition_of_transformation_logics_converse},
  we have that $\mathbf{S}(\mathcal{G})$ divides the flip-flop monoid $F$.
  We show this leads to a contradiction.
  We have that $\mathbf{S}(\mathcal{G})$ is a group $G$.
  By definition, when $G$ divides $F$, there exists a homomorphism $\psi$ from a
  subsemigroup of $F$ to $G$.
  Let $F = \{ a,b,e \}$ with $e$ the identity element.
  We have that $a \cdot s = s$ and $e \cdot a = a$.
  Thus $\psi(a) \cdot \psi(b) = \psi(a \cdot b) = \psi(b)$
  and
  $\psi(e) \cdot \psi(b) =  \psi(e \cdot b) = \psi(b)$.
  It implies that $\psi(b)$ does not admit an inverse in $G$, which is a
  contradiction since $G$ is a group. \todo{is it correct ??}
  We conclude that $P$ does not capture $\{ r \}$, and hence
  the expressivity of $\mathcal{L}(\mathcal{F}, \mathbf{G})$ properly includes
  the expressivity of $\mathcal{L}(\mathcal{F})$.
\end{proof}

\subsection{Proof of Theorem~\ref{theorem:expressivity_of_cyclic_operators}}

\theoremexpressivityofcyclicoperators*
\todo{there is a commented out proof idea}
\begin{proof}
  Consider a \wprogram{} $P$ of $\mathcal{L}(\mathcal{F}, \mathbf{S})$.
  Following the construction in the proof of
  Theorem~\ref{theorem:prime_decomposition_of_transformation_logics_converse},
  we can build an automaton $A$ that captures $P$ and whose semiautomaton is a
  cascade $C = (\phi_1, A_1) \ltimes \dots \ltimes (\phi_d, A_d)$ with $A_i$
  the flip-flop semiautomaton or a grouplike semiautomaton defined a solvable
  group $G \in \mathbf{S}$.
  In the latter case, we have that $G$ admits a composition series whose
  factors are cyclic groups of prime order. Let $\mathbf{C}_i$ denote such
  factors.
  By the Jordan–H\"{o}lder decomposition theorem
  (Theorem~\ref{theorem:jordan_holder}),
  $A_i$ is homomorphically represented by a semiautomata cascade 
  $(\phi_{i,1}, A_{i,1}) \ltimes \dots \ltimes (\phi_{i,1}, A_{i,n_1})$
  where every semiautomaton $A_{i,j}$ is defined by a cyclic group in
  $\mathcal{C}_i$.
  By replacing each such component in the cascade $C$, we obtain a cascade $C'$
  that homomorphically represents $C$ and where every semiautomaton is the
  flip-flop semiautomaton or it is defined by a cyclic group of prime order.
  By Proposition~1, there is an automaton $A'$ with semiautomaton $C'$ that is
  equivalent to $A$.
  Following the construction in the proof of
  Theorem~\ref{theorem:prime_decomposition_of_transformation_logics},
  we construct a \wprogram{} $P'$ that is equivalent to $A'$ and where every operator
  is a flip-flop operator or a cyclic operator of prime order. 
  In particular, we have that $P'$ captures $P$.
  Since $P$ is an arbitrary \wprogram{} of $\mathcal{L}(\mathcal{F}, \mathbf{S})$,
  we obtain that the expressivity of $\mathcal{L}(\mathcal{F}, \mathbf{S})$
  is contained in the expressivity of $\mathcal{L}(\mathcal{F}, \mathbf{C})$.
\end{proof}

\subsection{Proof of Theorem~\ref{theorem:expressivity_containment}}

\theoremexpressivitycontainment*
\begin{proof}
  We prove the theorem in two parts.

  First, we have that
  the expressivity of $\mathcal{L}(\mathcal{F},\mathbf{C}_1)$ equals the
  expressivity of $\mathcal{L}(\mathcal{F},\mathbf{C}_2)$ when  
  $\mathbf{C}_1 = \mathbf{C}_2$, since they are the same logic.

  For the second part, let us assume that $\mathbf{C}_1 \neq \mathbf{C}_2$.
  There exists an operator $\mathcal{C}_p$ for some prime number $p$ such that 
  $\mathbf{C}_1 \setminus \mathbf{C}_2$ or
  $\mathbf{C}_2 \setminus \mathbf{C}_1$.
  Let us assume the case $\mathbf{C}_2 \setminus \mathbf{C}_1$, and omit
  the proof of the other case since it will be symmetric.
  Specifically, we have $\mathcal{C}_p \in \mathbf{C}_2$ and
  $\mathcal{C}_p \notin \mathbf{C}_1$.
  Let $r$ be a \wrule{} of the form
  \begin{align*}
    p_1, \dots, p_n \ruleif \mathcal{C}_p(a_1, \dots, a_m \mid x_0).
  \end{align*}
  It suffices to show that there is no \wprogram{} of 
  $\mathcal{L}(\mathcal{F},\mathbf{C}_1)$ that captures $r$.
  Let us assume by contradiction that there is a \wprogram{} 
  $P'$ of $\mathcal{L}(\mathcal{F},\mathbf{C}_1)$ that captures $r$.
  By
  Theorem~\ref{theorem:prime_decomposition_of_transformation_logics_converse},
  it follows that the characteristic semigroup $C_p$ of $\mathcal{C}_p$ divides
  the flip-flop monoid $F$ or the characteristic semigroup of an operator in 
  $\mathbf{C}_1$.
  We exclude that $C_p$ divides $F$, as already argued in 
  the proof of Theorem~\ref{theorem:groups_add_expressivity}.
  Thus, let us consider the case where $C_p$ divides some cyclic group $C_q$,
  with $q \neq p$.
  We show it cannot be.

  There is a homomorphism from a subsemigroup $S$ of $C_q$ to $C_p$.
  It is a surjective function $\psi: S \to C_p$ such that $\psi(s_1 \cdot s_2)
  = \psi(s_1)  \cdot \psi(s_2)$ for every $s_1,s_2 \in S$.
  Let $g$ be a generator of $C_p$.
  Since $\psi$ is surjective, there exists $s \in S$ such that $\psi(s) = g$.
  We have that every element of $C_p$ is given by $g^n$ for $n \in [0, p]$.
  Equivalently it is of the form $\psi(s)^n = \psi(s) \cdot \psi(s) = \psi(s^n)$.
  Thus $s$ is the generator of a cyclic subgroup 
  $C \subseteq S \subseteq C_q$.
  The order of $C$ must divide the order of $C_q$, cf.\ Theorem~4.3 of
  \cite{gallian2021contemporary}.
  It follows that either $C = \{ e \}$ or $C = S = C_2$.
  In the first case, $C = \{ e \}$ implies that also $C_p = \{ e \}$, since 
  every element of $C_p$ is of the form $\psi(s^n) = \psi(e^n) = \psi(e) = e$.
  This is excluded because the order of $C_p$ is $p$, which is a prime number
  and hence it is different from one.
  In the second case, we have $C = S = C_q$.
  It follows that $C_p$ is a homomorphic image of $C_q$.
  Since $C_p$ is simple (i.e., it has no normal subgroup), it follows that 
  either $C_q$ is trivial or it is equal to $C_p$---see Page~7 of 
  \cite{domosi2005algebraic} or the end of Section~1.14 of \cite{ginzburg}.
  Both cases are excluded since we have assume $q \neq 1$ and $C_p \neq C_q$.
\end{proof}

\subsection{Proof of Theorem~\ref{theorem:non_canonicity}}

\theoremnoncanonicity*
\begin{proof}
  Let us consider $\mathbf{P}_1 = \{ \mathcal{F}, \mathcal{A}_5 \}$ and
  $\mathbf{P}_2 = \{ \mathcal{F}, \mathcal{A}_5,  \mathcal{C}_2 \}$ where
  $\mathcal{A}_5$ is the operator defined by the alternating group $A_5$, and
  $\mathcal{C}_2$ is the operator defined by the cyclic group of order two.
  Note that $\mathcal{C}_2$ and $\mathcal{A}_5$ are prime operators,
  since $C_2$ and $A_5$ are simple groups, cf.\
  \cite{gorenstein2018classification}.
  Note also that $C_2$ is a subgroup of $A_5$, cf.\ \cite{subgroupsofafive}.
  Namely, there is an isomorphism between a subgroup of $A_5$ and $C_2$,
  which implies that $C_2$ divides $A_5$.

  First, the expressivity of $\mathcal{L}(\mathbf{P}_1)$ is contained in 
  the expressivity of $\mathcal{L}(\mathbf{P}_2)$ since 
  $\mathbf{P}_1 \subset \mathbf{P}_1$.
  To show that the expressivity of $\mathcal{L}(\mathbf{P}_2)$ is contained in 
  the expressivity of $\mathcal{L}(\mathbf{P}_1)$, let us consider an arbitrary
  \wprogram{} $P$ of $\mathcal{L}(\mathbf{P}_2)$.
  It suffices to show that $P$ is captured by some \wprogram{} of 
  $\mathcal{L}(\mathbf{P}_1)$.
Following the construction of
  Theorem~\ref{theorem:prime_decomposition_of_transformation_logics_converse},
  we can give a cascade $A$ that captures $P$ where every semiautomaton is
  defined by a semigroup in $\{ F, A_5, C_2 \}$.
  Since $C_2$ divides $A_5$, we have that every semiautomaton define by $C_2$
  can be replaced with the semiautomaton defined by $A_5$, and obtain an
  equivalent cascade $A'$.
  Then, following the construction of
  Theorem~\ref{theorem:prime_decomposition_of_transformation_logics},
  from $A'$ we obtain a \wprogram{} that captures $P$ and has operators in
  $\mathcal{P}_1$.
\end{proof}

\subsection{Proof of Theorem~\ref{theorem:hierarchy}}

\theoremhierarchy*
\begin{proof}
  The strict containments when adding cyclic operators of prime order follow
  from Theorem~\ref{theorem:expressivity_containment}.
  The containment of the infinite hierarchy into 
  $\mathcal{L}(\mathcal{F}, \mathbf{S})$ follows from the fact that
  every cyclic group is solvable.
\end{proof}

\subsection{Proof of Theorem~\ref{theorem:complexity}}

\theoremcomplexity*
\begin{proof}
  For the first claim we provide an polynomial-time algorithm.
  Consider a query $(P,q)$ with $P \in \mathcal{L}$, and an input 
  $I = I_1, \dots, I_\ell$\/ for $P$.
  Let us recall that \wprogram{} $P$ has an associated acyclic dependency graph.
  In particular, there is one node for each variable, and each node has a depth,
  since the dependency graph is acyclic.
  The algorithm proceeds for increasing time $t$ and for increasing depth $n$,
  computing 
  (i) for all variables $p$ at depth $n$ whether they are true at time $t$, and
  (ii) for all transformation \wrule{}s $r$ at depth $n$ the associate element.
  Namely, it computes whether $(P,I,t) \models p$ holds, and the element $x$
  such that $(P,I,t) \models (r \mapsto x)$.
For $n=0$, variable $p$ is an input variable, and to check
  $(P,I,t) \models p$ it suffices to check $p \in I_t$.
  Let us consider $n>0$ and $t \geq 0$.
  We consider three separate cases, according to whether $p$ is defined by a
  static \wrule{}, a delay \wrule{}, or a transformation \wrule{}.
  In the case of a static \wrule{}, the algorithm has already computed the truth
  value at time $t$ for the variable in the \wbody\ of the \wrule{}.
  To determine the value of $p$, it suffices to evaluate the propositional
  formula in the \wbody{}, which can be done in polynomial time.
  In the case of a delay \wrule{}, if $t = 1$, the algorithm determines that the
  truth value of $p$ at time $t$ is false, and if $t \geq 1$, the algorithm
  determines the truth value of $p$ at time $t$ as the truth value of the
  \wbody\
  variable at time $t-1$, which has already been computed.
  In the case of a transformation \wrule{}, 
  the element $x$ is such that $(P,I,t) \models (r \mapsto x)$
  is given by $x = \tau(x')$ for $\tau = \phi(\mu)$ and 
  $(P,I,t) \models (a_1,\dots,a_m) \to \mu$ and 
  $(P,I,t-1) \models (r \mapsto x')$; we have that the truth of variables $a_i$
  has already been computed, as well as the value $x'$; the functions $\phi$ and
  $\tau$ can be evaluated in polynomial time since all operators are polytime.
  Furthermore, the value of the variable $p$ defined by the \wrule{} is given by
  $b_i$ as determined by $\psi(x) = \langle b_1, \dots, b_n \rangle$. We have
  that $\psi$ can be computed in polynomial time since all operators are
  polytime.
  Since there are at most $|P|$ \wrule{}s,
  the number of steps is $O(\ell \cdot |P| \leq |I| \cdot |P|)$.

  For the second claim, consider the family of operators 
  $\mathcal{H}_n = \langle \mathbb{B}, T, \phi, \mathit{id} \rangle$ where 
  the set $T$ of transformations is $\mathit{set}(x) = 1$
  and $\mathit{read}(x) = x$,
  and the function
  $\phi: \mathbb{B}^n \to \mathbb{B}$ returns $1$ iff
  its input is the binary encoding of a positive instance of the propositional
  logic programming of length $n$. 
  Note that propositional logic programming is \textsc{PTime}-complete, cf.\
  \cite{dantsin2001complexity}.
  Let $\mathbf{H} = \{ \mathcal{H}_n \mid n \geq 1 \}$.
  We show a logspace reduction $\psi$ from propositional logic programming to
  the evaluation problem of the logic $\mathcal{L}(\mathbf{H})$.
  This shows that the problem is \textsc{PTime}-hard, since 
  propositional logic programming is \textsc{PTime}-complete, cf.\
  \cite{dantsin2001complexity}.
  Consider a family of singleton \wprogram{}s $P_n$ consisting of the \wrule{}
  \begin{align*}
    q \ruleif \mathcal{H}_n(a_1, \dots, a_n).
  \end{align*}
  The reduction consists in mapping an instance $I$ of logic
  programming to the evaluation problem with input 
  $\psi(I) = \langle (P,q), I, 1 \rangle$.
  We have that $I$ is a positive instance of logic programming iff $\psi(I)$ is
  a positive instance of the evaluation problem of $\mathcal{L}(\mathbf{H})$.
  In fact, we have that $(P, I, 1) \models q$ iff $I$ is a positive instance of
  logic programming.
\end{proof}

\subsection{Proof of Lemma~\ref{lemma:complexity_finite_operators}}

\lemmacomplexityfiniteoperators*
\begin{proof}
  The two claims are proved separately.

  For the first claim.
  we have that every finite set of finite operators can be evaluated in constant
  time $O(1)$ since the size of an operator is bounded by a constant.
  
  For the second claim,
  consider the family of operators 
  $\mathcal{H}_n = \langle \mathbb{B}, T, \phi, \mathit{id} \rangle$ where 
  the set $T$ of transformations is $\mathit{set}(x) = 1$
  and $\mathit{read}(x) = x$,
  and the function
  $\phi: \mathbb{B}^n \to T$ returns $\mathit{set}$ iff
  its input is the binary encoding of a positive instance of the Datalog
  evaluation problem of length $n$. 
  Since Datalog is \textsc{ExpTime}-complete, cf.\ \cite{dantsin2001complexity},
  $\textsc{PTime} \subsetneq \textsc{ExpTime}$ by the time hierarchy theorem,
  and evaluating $\phi$ is necessary to determine the value of 
  $\psi(\tau(x))$ with $\tau = \phi(\mu)$ when given $x$ and $\mu$,
  we conclude that the family of operators $\mathcal{H}_n$ is not polytime.
\end{proof}

\subsection{Proof of Theorem~\ref{theorem:complexity_finite_operators}}

\theoremcomplexityfiniteoperators*
\begin{proof}
  By Lemma~\ref{lemma:complexity_finite_operators}, we have that $\mathbf{T}$ is
  polytime.
  Then the theorem follows by Theorem~\ref{theorem:complexity}.
\end{proof}

\subsection{Proof of Proposition~\ref{proposition:compact_operators}}

\propositioncompactoperators*
\begin{proof}
  Every operator $\mathcal{T}_n$ can be represented by encoding the symbol
  $\mathcal{T}$ with a constant number of bits and the index $n$ in binary with
  a number of bits that is $O(\log n)$. 
  The same argument applies to every operator $\mathcal{C}_n$.
\end{proof}

\subsection{Proof of Lemma~\ref{lemma:complexity_counters}}

\lemmacomplexitycounters*
\begin{proof}
  We deal with the two families of operators separately.
  We first discuss $\{ \mathcal{T}_n \}$.
  Let $m$ the minimum number of bits to encode $n$. Note that $m \in O(\log n)$.
We recall the definition of $\mathcal{T}_n$.
  We have $\mathcal{T}_n = \langle [0,n], T, \phi, \mathit{id} \rangle$
  where $T$ and $\phi$ are defined as follow.
  First, we define $T = \{ \mathit{inc}, \mathit{id} \}$ with $\mathit{inc}$
  defined as $\mathit{inc}(x) = \min(n, x + 1)$ and $\mathit{id}$ the identity
  function.
  Second, we define $\phi$ as $\phi(0) = \mathit{id}$ and 
  $\phi(1) = \mathit{inc}$.
  We have that the function $\phi$ is constant-time and all the other
  functions can operate in polynomial time over
  the binary representation of elements $x \in [0,n]$, and hence in time
  polynomial in the size of the operator, as required.

  Next we discuss $\{ \mathcal{C}_n \}$.
  We provide a definition of $\mathcal{C}_n$ which makes it clear that the
  operator is polytime.
  Let $m$ the minimum number of bits to encode $n-1$. 
  Note that $m \in O(\log n)$.
  Consider the function $\varphi_n: \mathbb{B}^m \to [0,n-1]$ defined as 
  \begin{align*}
    \varphi_n(b_1, \dots, b_m) = \min\left(n-1,\, \sum_{i=0}^{m-1} b_i \cdot 2^i
    \right)
  \end{align*}
  We can define $\mathcal{C}_n$ as 
  $\langle [0,n-1], T, \varphi_n, \mathit{id} \rangle$
  where $T = [0,n-1]$ with $\tau \in T$ defined as 
  $\tau(x) = x + \tau \operatorname{mod} n$.
  We have that the function $\varphi$ is polytime in $m$ and hence in the
  size of the operator, and that all the other functions can operate in
  polynomial time over the binary representation of elements $x \in [0,n]$, and
  hence in time polynomial in the size of the operator, as required.
\end{proof}

\subsection{Proof of Theorem~\ref{theorem:overall_complexity}}

\theoremoverallcomplexity*
\begin{proof}
  The families $\{ \mathcal{T}_n \}$ and $\{ \mathcal{C}_n \}$ are polytime
  by Lemma~\ref{lemma:complexity_counters}.
  The family $\mathbf{A}$ is polytime by
  Lemma~\ref{lemma:complexity_finite_operators}.
  We conclude that the family 
  $\mathbf{A} \cup \{ \mathcal{T}_n \} \cup \{ \mathcal{C}_n \}$ is
  polytime.
  Then, the theorem follows by Theorem~\ref{theorem:complexity}.
\end{proof}

\subsection{Proof of Theorem~\ref{theorem:constant_depth_flipflop}}

\todo{For all the proofs below, point out that the encoding/decoding functions
can be captured by AC0 circuits according to the following reasoning. 
The encoding/decoding functions operate over a fixed number of bits since the
set of operators is fixed. Hence such functions are captured by Boolean circuits
of constant size. Boolean circuits of constant size are clearly in AC0.}

\theoremconstantdepthflipflop*
\begin{proof}
  Given a query $(P,q)$ with $P \in \mathcal{L}(\mathcal{F}, k)$, it suffices to
  show a family of $\textsc{AC}^0$ circuits that computes $(P,I,\ell) \models q$
  for every input $I = I_1, \dots, I_\ell$ to $P$. Let us recall that an
  $\textsc{AC}^0$ circuit has `not' gates, unbounded fan-in `and' gates,
  unbounded fan-in `or' gates, and it has constant depth and polynomial size.
  For transformation \wrule{}s with flip-flop operator,
  an existing construction in the proof of Theorem~2.6 of
  \cite{chandra1985unbounded} shows how to construct an $\textsc{AC}^0$ circuit
  that computes the bit stored by a flip-flop at time $t$.
  Static \wrule{}s also correspond to $\textsc{AC}^0$ circuits since the
  \wprogram{} has
  constant-depth, and delay \wrule{}s correspond to the flip-flop that sets when
  the input is true and resets otherwise (i.e., it stores the previous bit), and
  hence admit the construction above.  
  Thus, considering that \wprogram{} $P$ has constant depth, 
  we can combine the aforementioned circuits for $t$ from $1$ to $\ell$ and
  obtain a circuit of constant depth and polynomial size whose output is $1$ on
  input $I$ if and only if $(P,I,\ell) \models q$.
  We provide a more detailed construction of the circuit below.

  \par\smallskip\noindent
  \emph{Details of the circuit.}
  Let $a_1, \dots, a_n$ be the input variables of the given 
  \wprogram{} $P$.
  The circuit has one input pin for the truth assignment of every variable
  at every time point $t \in [1,\ell]$ as specified by the input 
  $I = I_1, \dots, I_\ell$, hence it has input pins 
  $a_{1,1}, \dots, a_{n,\ell}$.
  The construction proceeds by introducing gates in order to ensure that,
  for every variable $p$ defined in $P$ and every time point $t \in
  [1,p]$, there is a gate $g_{p,t}$ whose output corresponds to 
  $(P,I,t) \models p$. 
  We proceed by considering a \wrule{} $r$ of $P$ such
  that all variables occurring in the \wbody\ of $r$ have their corresponding
  gates already included in the circuit---since programs have an acyclic
  dependency graph, such a \wrule{} always exists unless the circuit is already
  complete. 
  We consider three cases separately, according to what kind of \wrule{}
  $r$ is.

  First, let us consider the case when $r$ is a transformation \wrule{}.
  We have that $r$ is a transformation \wrule{} with flip-flop operator. Let it
  be of the following form.
  \begin{align*}
    p \ruleif \mathcal{F}(a,b \mid x_0)
  \end{align*}
  Here we adopt the construction from Theorem~2.6 of
  \cite{chandra1985unbounded}.
  We consider two cases separately.
  In the first case we have $x_0 = 0$.
  We extend the circuit built so far with the following circuits
  for $t$ ranging from $1$ to $\ell$.
  \begin{align*}
  g_{p,t} = \bigvee_{1 \leq i \leq t} 
  \left( g_{a,i} \land \bigwedge_{i < j \leq t} \neg g_{b,j} \right)
  \end{align*}
  Note that every $g_{a,i}$ and every $g_{b,j}$ is a gate already present in
  the circuit built so far.
  Intuitively, $p$ is true at time $t$ if there is a point $i$ at which the
  flip-flop has been set ($a$ is true) and it has not been reset at any of the
  following time points ($b$ is false).
  The case $x_0 = 1$ is handled similarly, except that the circuits to add are
  as follows:
  \begin{align*}
    g_{p,t} = \left( \bigwedge_{0 < j \leq t} \neg g_{b,j} \right) \lor
    \bigvee_{1 \leq i \leq t} \left( g_{a,i} \land \bigwedge_{i < j \leq t} \neg
    g_{b,j} \right)
  \end{align*}

  Second, let us consider the case when $r$ is a delay \wrule{}.
  Let the delay \wrule{} be of the following form.
  \begin{align*}
    p \ruleif \operatorname{\mathcal{D}} a
  \end{align*}
  We extend the circuit built so far with the following circuit where
  $\bot$ always evaluates to false,
  \begin{align*}
    g_{p,1} = \bot,
  \end{align*}
  and the following circuits for $t$ ranging from $2$ to $\ell$,
  \begin{align*}
    g_{p,t} = g_{a,{t-1}}.
  \end{align*}
  Intuitively, $p$ is false at time $1$, and its truth corresponds to the one of
  $a$ at the previous time point for time points in $[2,\ell]$.
  
  Third, let us consider the case when $r$ is a static \wrule{}.
  Let $b_1, \dots, b_m$ be the variables in the \wbody\ of $r$, and let $p$ be
  the variable defined by $r$.
  The \wbody\ of $r$ provides us with a circuit which we can copy $\ell$ times,
  with the $t$-th copy having inputs $g_{b_1,t}, \dots, g_{b_m,t}$. Then
  the $t$-th circuit has a gate that defines $g_{p,t}$.
\end{proof}

\subsection{Proof of Theorem~\ref{theorem:constant_depth_acc}}

\begin{definition}
  Let $S$ be a semigroup. The word problem for $S$ is the problem, for given
  $s_1, s_2, \dots, s_n \in S$, of computing the element
  $s_1 \cdot s_2 \cdots s_n$.
\end{definition}

The following theorem is Theorem~5 of \cite{barrington1989bounded}.
\begin{theorem}[Barrington, 1989]
  \label{theorem:reduction_of_solvable_groups}
  The word problem for any fixed solvable group $G$ is $\textsc{AC}^0$-reducible
  to the function $f_g(x) = x \operatorname{mod} g$, where $g$ is the
  order of $G$.
\end{theorem}

\theoremconstantdepthacc*
\begin{proof}
  We prove the theorem in two separate parts.

  For the first part, given a query $(P,q)$ with $P \in \mathcal{L}(\mathcal{F},
  \mathbf{S}, k)$, it suffices to
  show a family of $\textsc{ACC}^0$ circuits that computes $(P,I,t) \models q$
  for every input $I$ to $P$ and every time point $t$. 
  Let us recall that an $\textsc{ACC}^0$ circuit is an $\textsc{AC}^0$ circuit
  where also unbounded fan-in modulo-$g$ gates are allowed for any $g$, which
  compute the sum of their inputs modulo $g$.
  Using Theorem~\ref{theorem:reduction_of_solvable_groups}, 
  every transformation \wrule{} with solvable group operators is captured by
  an $\textsc{ACC}^0$ circuit.
  The other transformation \wrule{}s and the rest of the argument follow the
  proof of Theorem~\ref{theorem:constant_depth_flipflop}.
  We provide further details on the construction of the circuit below, after
  the proof for the second part.

  For the second part, let us consider the parity function 
  $f : \mathbb{B}^* \to \mathbb{B}$ that returns the parity of its input bits,
  i.e., their sum modulo 2.
  Let us consider the operator $\mathcal{C}_2$ defined by the cyclic group of
  order 2, and the \wrule{} $r$ defined as
  \begin{align*}
    p \ruleif \mathcal{C}_2(a \mid 1).
  \end{align*}
  Let $P = \{ r \}$.
  We have that, for an input $I$ to $P$, it holds that 
  $(\{r\}, I, t) \models p$ iff $a$ occurs an even number of times in $I$.
  Namely, \wprogram{} $P$ captures parity.
  Since parity is not in $\textsc{AC}^0$ \cite{furst1984parity}, 
  and $P \in \mathcal{L}(\mathcal{F}, \mathcal{C}_2, 1)$,
  we conclude that
  the evaluation problem of
  $\mathcal{L}(\mathcal{F}, \mathcal{C}_2, k)$ is not in $\textsc{AC}^0$.
  Note that $\mathcal{C}_2$ is a solvable group operator as required, since
  every cyclic group is solvable.

  \par\medskip\noindent
  \emph{Details of the circuit.}
  The circuit for the first part of the proof is built as in the proof of
  Theorem~\ref{theorem:constant_depth_flipflop} except for the case of a
  transformation \wrule{} featuring a group operator.
  We focus on this case next.
  Let us consider a solvable group $G = \langle X, \cdot \rangle$, 
  its corresponding solvable group operator 
  $\mathcal{G} = \langle X, T, \phi, \psi \rangle$,
  and a transformation \wrule{} featuring the operator.
  \begin{align*}
    p_1, \dots, p_n \ruleif \operatorname{\mathcal{G}}(q_1, \dots, q_m \mid x_0)
  \end{align*}
Theorem~\ref{theorem:reduction_of_solvable_groups} of
  \cite{barrington1989bounded} says that, given elements
  $x_0, x_1, \dots, x_t$ of $G$, one can compute the binary encoding of the 
  product element $x_0 \cdot x_1 \cdots x_t$ starting from the binary encodings
  of $x_0, x_1, \dots, x_t$ and using $\textsc{AC}^0$ circuitry and a
  modulo-$|X|$ gate. Note that \cite{barrington1989bounded} assumes fixed
  encodings, which is true in our case since every operator has its own encoding
  as determined by $\phi$ and $\psi$, and the set $\mathbf{G}$ of operators is
  fixed as well.
  Then we have that the circuit built so far has gates 
  $g_{q_1, t}, \dots, g_{q_m, t}$, for $t \in [1,\ell]$, that provide us a
  binary encoding of the element of $X$ that is associated with the
  transformation \wrule{} at time $t$.
  We add constant gates $g_{q_1, 0}, \dots, g_{q_m, 0}$ that provide us with
  the binary encoding of the initial element $x_0$.
  By Theorem~\ref{theorem:reduction_of_solvable_groups}, we can use
  $\textsc{AC}^0$ circuitry and a modulo-$|X|$ gate to obtain gates 
  $g_{p_1, t}, \dots, g_{p_n, t}$ for every $t \in [1,\ell]$, which are the
  required gates for variables $p_1, \dots, p_n$.
\end{proof}

\subsection{Proof of Theorem~\ref{theorem:complexity_nc}}

The following theorem is Theorem~4 of \cite{barrington1989bounded}.
\begin{theorem}[Barrington, 1989]
  \label{theorem:reduction_of_nonsolvable_groups}
  The word problem of any fixed non-solvable group is complete for
  $\textsc{NC}^1$ under $\textsc{AC}^0$ reductions.
\end{theorem}

\theoremcomplexitync*
\begin{proof}
  For membership in $\textsc{NC}^1$,
  by Theorem~\ref{theorem:reduction_of_nonsolvable_groups},
  we have that transformation \wrule{}s with non-solvable groups can be reduced to
  $\textsc{NC}^1$.
  Transformation \wrule{}s with solvable groups and flip-flops, as well as static
  and delay \wrule{}s can be reduced to $\textsc{ACC}^0 \subseteq \textsc{NC}^1$ as
  discussed in the proofs of Theorem~\ref{theorem:constant_depth_acc} and
  Theorem~\ref{theorem:constant_depth_flipflop}.

  The converse holds since we can $\textsc{AC}^0$-reduced every 
  $\textsc{NC}^1$ circuit to the word problem of a group
  \cite{barrington1989bounded},
  and we can express both the word problem and any 
  $\textsc{AC}^0$ circuit as a \wprogram{} of 
  $\mathcal{L}(\mathcal{F},\mathbf{G},k)$ for some $\mathbf{G}$ and some $k$.
\end{proof}

\subsection{Proof of Lemma~\ref{lemma:pltl_beforeop}}

\lemmapltlbeforeop*
\begin{proof}
  It follows from the fact that the semantics of delay \wrule{}s coincides with the
  semantics of the before operator.
\end{proof}

\subsection{Proof of Lemma~\ref{lemma:pltl_sinceop}}

\lemmapltlsinceop*
\begin{proof}[Proof idea]
  By induction on $t$, making use of the alternative, recursive definition of
  the semantics of the `since' operator.
\end{proof}
\begin{proof}
  Let $r_1$ be the \wrule{}
  \begin{align*}
    c \ruleif \neg a,
  \end{align*}
  and let $r_2$ be the \wrule{}
  \begin{align*}
    p \ruleif \mathcal{F}(b, c \mid 0).
  \end{align*}
  Let $I = I_1, \dots, I_\ell$.
  We proceed by induction on $t$ from $1$ to $\ell$.

  In the base case $t = 1$.
  We show the two implications separately.
  First, we show that
  $(I,1) \models (a \sinceop b)$ implies $(P,I,1) \models p$.
  We have that $(I,1) \models (a \sinceop b)$ implies $b \in I_1$, hence
  $(P,I,1) \models (r_2 \mapsto 1)$, and hence $(P,I,1) \models p$ as required.
  Second, we show that
  $(P,I,1) \models p$ implies $(I,1) \models (a \sinceop b)$.
  We have that 
  $(P,I,1) \models p$ implies (Point~5 of the semantics) that 
  $(P,I,1) \models (r_2 \mapsto 1)$, which in turn implies
  that (Point~6 of the semantics and definition of $\mathcal{F}$)
  $(P,I,1) \models b$, hence
  $b \in \sigma_1$, and hence $(I,1) \models (a \sinceop b)$.
  
  In the inductive case $t \geq 2$, and we assume 
  that $(I,t-1) \models (a \sinceop b)$ iff $(P,I,t-1) \models p$.
  We show the two implications separately.
  First, we show that
  $(I,t) \models (a \sinceop b)$ implies $(P,I,t) \models p$.
  We have that $(I,t) \models (a \sinceop b)$ implies 
  (i) $(I,t) \models b$ or 
  (ii) $(I,t-1) \models (a \sinceop b)$ and $(I,t) \models a$.
  The first case is analogous to the one in the base case.
  Let us consider the second case, by assuming 
  $(I,t-1) \models (a \sinceop b)$ and $(I,t) \models a$.
  Hence, $a \in I_t$, and by the inductive hypothesis
  $(P,I,t-1) \models p$.
  We have that 
  $(P,I,t-1) \models p$ implies that
  $(P,I,t-1) \models (r_2 \mapsto 1)$,
  and hence it suffices to show 
  $(P,I,t) \not\models c$, which holds since 
  $a \in I_t$.
  Now we show the other implication, i.e., 
  $(P,I,t) \models p$ implies $(I,t) \models (a \sinceop b)$.
  We assume $(P,I,t) \models p$.
  It follows that $(P,I,t) \models (r_2 \mapsto 1)$.
  By the definition of $\mathcal{F}$,
  it implies that one of the two following conditions holds:
  (i) $(P,I,t) \models b$ and hence $b \in I_t$, or
  (ii) $(P,I,t-1) \models (r_2 \mapsto 1)$ and $(P,I,t) \not\models c$ (which is
  equivalent to $a \in I_t$).
  We consider the two possible cases separately.
In the first case $b \in I_t$ implies $(I,t) \models (a \sinceop b)$.
  In the second case $(P,I,t-1) \models p$.
  By the inductive hypothesis we have 
  $(I,t-1) \models (a \sinceop b)$, which implies
  $(I,t) \models (a \sinceop b)$ since 
  $a \in I_t$.

  This concludes the proof.
\end{proof}

\subsection{Proof of Theorem~\ref{theorem:pltl_translation}}

\theorempltltranslation*
\begin{proof}
  Consider a Past LTL formula $\varphi$.
  We show a query $(P,q)$ such that, 
  for every interpretation $I$ of $\varphi$ and every time point $t$,
  it holds that $(I, t) \models \varphi$ iff 
  $(P, I, t) \models q$.
  We proceed by induction on the depth of a parse-tree of $\varphi$.
  In the base case the depth is zero, the parse-tree consists of the root only,
  and $\varphi = a$ where $a$ is a propositional variable.
  The corresponding query is $(P,q)$ where $q$ is a fresh variable and $P$
  consists of the \wrule{} $q \ruleif a$.

  In the inductive case the depth $d$ of a parse-tree of $\varphi$ is greater
  than zero, and we assume that we can build a \wprogram{} for every formula whose
  parse-tree has depth at most $d-1$.
  We proceed by cases, according to the top-level operator of $\varphi$, the one
  associated with the root.
  In the case of $\lor$, 
  given the queries $(P_\alpha, a_\alpha)$ and $(P_\beta, a_\beta)$ for the
  formulas $\alpha$ and $\beta$ associated with the children of the root, 
  the corresponding query is $(P,q)$ where $q$ is a fresh variable and $P$ is 
  $P_\alpha \cup P_\beta$ extended with the \wrule{}
  $q \ruleif a_\alpha \lor a_\beta$.
  Similarly for $\land$.
  In the case of $\neg$, given the query $(P_\alpha, a_\alpha)$ for the formula
  $\alpha$ associated with the child of the root, and the resulting query is
  $(P,q)$ where $q$ is a fresh variable and $P$ is $P_\alpha$ extended with the
  \wrule{} $q \ruleif \neg a_\alpha$.
  In the case of $\beforeop$, given the query $(P_\alpha, a_\alpha)$ for the
  formula $\alpha$ associated with the child of the root, the resulting query is
  $(P,q)$ where $q$ is a fresh variable and $P$ is $P_\alpha$ extended with the
  \wrule{} $q \ruleif \operatorname{\mathcal{D}} a_\alpha$.
  In the case of $\sinceop$, given the queries
  $(P_\alpha, a_\alpha)$ and $(P_\beta, a_\beta)$ for the formulas $\alpha$ and
  $\beta$ associated
  with the children of the root, the resulting query is $(P,q)$ where $q$ is a
  fresh variable and $P$ is $P_\alpha \cup P_\beta$ extended with the \wrule{}s
  $q \ruleif \mathcal{F}(a_\beta, c \mid 0)$ and $c \ruleif \neg a_\alpha$.

  Correctness of the construction for the Boolean operators follows from the
  fact that their semantics in the Transformation Logics is the standard one,
  and hence it coincides with their semantics in Past LTL.
  Correctness for the `before' operator is by Lemma~\ref{lemma:pltl_beforeop}.
  Correctness for the `since' operator is by Lemma~\ref{lemma:pltl_sinceop}.
The size of the resulting \wprogram{} is linear since we introduce one \wrule{} for
  every operator of $\varphi$, and every \wrule{} has a constant number of
  variables.
The translation requires only to keep a constant number of pointers and
  variable names at the time, and hence it can be carried out in logarithmic
  space.
\end{proof}

 \fi

\end{document}